\documentclass{article}
\usepackage{geometry}
\usepackage{booktabs}
\usepackage{graphicx}
\usepackage{graphics}
\usepackage{lscape}
\usepackage{indentfirst}
\usepackage{latexsym}
\usepackage{authblk}
\usepackage{blindtext}
\usepackage{amsmath}
\usepackage{amssymb}
\usepackage{amsthm}
\usepackage{hyperref}
\usepackage{algorithm}
\usepackage{algpseudocode} 
\usepackage{color}
\usepackage{amssymb}
\usepackage{amsfonts}
\usepackage{amsmath}
\usepackage{bm}
\usepackage{mathrsfs}
\usepackage{braket}
\usepackage{siunitx}
\usepackage{subcaption} 

\newcommand{\bal}[1]{
\begin{align} #1 \end{align}
}

\newcommand{\baln}[1]{
\begin{align*} #1 \end{align*}
}

\newcommand{\del}[3]{\left#1 #3 \right#2}
\newcommand{\avv}[1]{\del{<}{>}{#1}}

\newtheorem{lemma}{Lemma}

\def\cH{ {\cal H} }

\def\cN{ {\cal N} }

\def\cS{ {\cal S} }

\def\cO{ {\cal O} }

\def\complex{\mathbb{C}}

\def\im{\text{Im}}
\def\re{\text{Re}}
\def\epsilon{\varepsilon}

\def\polyn{\text{poly}(n)}

\newcommand{\be}{\begin{equation}}
\newcommand{\ee}{\end{equation}}

\newcommand{\lam}{\lambda}

\newcommand{\bigo}[1]{{\cal O}\left({#1}\right)}

\newcommand{\cref}[1]{(\ref{#1})}

    \title{Why Adiabatic Quantum Annealing is unlikely to yield speed-up}
\author{Aar\'on Villanueva\footnote{aaronv@science.ru.nl}, Peyman Najafi, Hilbert J. Kappen}
\affil{Radboud University, Nijmegen, The Netherlands}

\begin{document}

\maketitle
\begin{abstract}
We study quantum annealing for combinatorial optimization with Hamiltonian $H = H_0 + z H_f$ where $H_f$ is diagonal, $H_0=-\ket{\phi}\bra{\phi}$ is the equal superposition state projector and $z$ the annealing parameter.
We analytically compute the minimal spectral gap, which is $\bigo{1/\sqrt{N}}$ with $N$ the total number of states, and its location $z_*$.
We show that quantum speed-up requires an annealing schedule which demands a precise knowledge of $z_*$, which can be computed only if the density of states of the optimization problem is known.
However, in general the density of states is intractable to compute, making quadratic speed-up unfeasible for any practical combinatorial optimization problems. 
We conjecture that it is likely that this negative result also applies for any other instance independent transverse Hamiltonians such as $H_0 = -\sum_{i=1}^n \sigma_i^x$.
\end{abstract}


\section{Introduction}

Combinatorial optimization problems emerge in many real world scenarios and solving them efficiently is of great interest and importance in science and engineering \cite{gar79}.
Simulated Annealing (SA) is a well-known heuristic classical algorithm,  which is general purpose and easy to implement \cite{kir83}.
The ergodicity of the Markov chain (thermal fluctuations) combined with a sufficiently slow annealing schedule, ensures that SA converges asymptotically to an optimal solution of the optimization problem. 
In practice, the total time and the efficiency of SA is determined by the mixing time of the Markov chain, which diverges for hard optimization problems
\cite{haj88}.

The idea of quantum annealing (QA) \cite{brooke1999quantum,kadowaki1998quantum,farhi2000quantum} was introduced motivated by the conjecture that quantum fluctuations can be more efficient than classical thermal fluctuations
 to solve hard optimization problems \cite{apolloni1989quantum, apolloni1990numerical}.
 This would constitute in a quantum advantage over classical algorithms.
Today, a well known example of quantum advantage is unstructured (Grover) search of one out of $N$ items which can be done in $\cO(\sqrt{N})$ time using a quantum computer \cite{grover1997quantum}.
It was shown in \cite{roland2002quantum} that the same quadratic speed-up can be obtained using quantum annealing (QA).

NP-complete problems are defined by an universality property: the existence of an efficient algorithm to solve them would solve any other hard problem instantly.
The question of whether QA can be used to solve efficiently NP-complete problems is a difficult topic that has remained open.
An archetypal example is  3-satisfiability (3-SAT). 
Early applications of QA to 3-SAT indicated speed-up over classical algorithms~\cite{young2008size, farhi2001quantum, hogg2003adiabatic, farhi2000quantum, schutzhold2006adiabatic}. However, such studies were restricted either using a small number of qubits or considering only a subset of typical instances while ignoring worst cases \cite{young2010first}. 

The question was also studied using Simulated QA, which uses quantum Monte Carlo. Initially, a better scaling of Simulated QA relative to SA 
was reported in~\cite{martovnak2002quantum,santoro2002theory} to find the lowest energy configuration for the  2d random Ising model and in~\cite{battaglia2005optimization} for random 3-SAT. 
However, it was reported later that this advantage was because of time discretization (Trotterization) artifacts, and in the continuous time limit  there is no superiority for QA~\cite{heim2015quantum}. Around the same time, it was argued that QA is exponentially slow for solving NP-hard problems due to exponential closing of the spectral gap with increasing system size~\cite{jorg2008simple,altshuler2009adiabatic}.

One can build examples in which QA is exponentially faster than SA. For instance, the ``Hamming weight with a spike" problem~\cite{farhi2002quantum,kong2017performance}, which has a high thin barrier that makes SA get stuck in a local minimum. It is designed to make the SA algorithm fail while QA can tunnel through the spike and is effective.
Nevertheless, this problem is not NP-hard and can be efficiently solved by other classical algorithms~\cite{muthukrishnan2016tunneling, crosson2016simulated}. 

Further numerical studies show that adiabatic QA performs worse than SA on 3-SAT 
\cite{neuhaus2011classical} and the antiferromagnet spin model on a 3-regular graph \cite{liu2015quantum}. 
An experimental comparison of SA and QA  on the D-Wave \cite{BIAN2020104609,kowalsky20223} quantum annealing for random spin glass instances also reported no quantum speed-up~\cite{ronnow2014defining}.

Based on these negative results, efforts have been made to improve QA by instance specific initialization~\cite{perdomo2011study}, the use of a more informed choice of the initial Hamiltonian~\cite{farhi2008make}, the use  of an instance specific biasing $\sigma^z$ field~\cite{grass2019quantum}, other strategies which might
allow an escape from bottlenecks or trap states~\cite{amin2008effect},
 increase of the dimensionality of the Hamiltonian parametrization~\cite{Rezakhani2009} or change the adiabatic path by adding a catalyst Hamiltonian which vanishes at the edges of the annealing~\cite{Farhi2002a}.
Although these techniques exhibited cases of success in simple or ad-hoc problems, so far this has not yielded a demonstration of quantum speed for larger instances.

Due to the uncertainty surrounding adiabatic methods, in recent years there has been explored other strategies such as non-adiabatic or {\it diabatic} quantum computing \cite{yan2022analytical, crosson2014different, crosson2021prospects, shi2020efficient, muthukrishnan2016tunneling}.
In the diabatic case, transitions between the ground and excited states are allowed in regions where the gap becomes exponentially small, therefore avoiding the time complexity carried by adiabatic evolution.
An example is the glued-tree problem~\cite{childs2003exponential}, for which there exists a mixed adiabatic and diabatic evolution schedule that offers to date the only {\it provable}~\cite{ronnow2014defining} quantum (exponential) speedup over classical algorithms~\cite{albash2018adiabatic}.
It is worth noticing that a superpolynomial quantum speed-up was proven in \cite{hastings2021} for stoquastic adiabatic computation in the oracle setting.
More recently, following the same lines of \cite{hastings2021}, in \cite{vazirani2021sub} a sub-exponential quantum speed-up was proven in the stoquastic adiabatic framework for a problem in close relation with glued-trees.

While the removal of the adiabatic condition gives more freedom in building quantum algorithms and has proven advantage in some cases, 
the question whether diabatic methods can solve NP-hard problems in general is still unanswered.

Even in the adiabatic case, the question whether it can give quantum speedup for hard combinatorial optimization problems remains an open question to date~\cite{albash2018adiabatic}.
Apart from some specific simplified cases, such as Grover search \cite{roland2002quantum}, there is no convincing demonstration of quantum advantage using adiabatic methods for generic hard problems.

It was shown for instance that QA can not give better than quadratic speed-up for a very general model Hamiltonian with linear schedule~\cite{vznidarivc2006exponential,farhi2008make}.
Since this result is a lower bound on the time complexity, it leaves open the question of whether quadratic speed-up using QA
for generic combinatorial optimization problems
is possible at all.
In this paper we explore this question and respond positively.
We study a model Hamiltonian 
composed of a diagonal term plus the equal superposition state projector. We compute the minimal spectral gap and its location $z_*$ in the annealing interval.
We provide an analytical expression for the spectral gap, which is valid in a vicinity of the minimal value.
We  prove that the minimal gap is $\bigo{1/\sqrt{N}}$
\footnote{We make the convention of omitting non-exponential corrections in complexity estimations, for instance $\bigo{\text{poly}(n)}$ or $\bigo{\log{n}}$ factors.}
which happens at $z_* = Z_1$ with $Z_1$ a function of the density of states.
When the density of states is known, we can design a schedule that achieves $T = \bigo{\sqrt{N}}$, i.e. quadratic speedup over naive search for any combinatorial optimization problem.
This generalizes earlier results in the case of Grover search~\cite{roland2002quantum}.
Furthermore, we prove that this complexity is optimal, generalizing a previous result that holds in the linear regime~\cite{farhi2008make} to the case of arbitrary non-linear schedules.

However, the annealing schedule depends sensitively on the precise location of the gap $z_* = Z_1$.
In general the density of states is instance dependent and is intractable to compute.
Therefore $Z_1$ cannot be computed and we cannot build the optimized schedule.
This issue was explored in \cite{slutskii2019analog} in the context of adiabatic unstructured search, leading to similar conclusions about the sensibility of the annealing schedule.
We argue that our result is likely to hold also for the more commonly used mixing Hamiltonian $H_0 = -\sum_{i=1}^{n} \sigma_i^x$ and we support this claim with numerical results (see Sec.~\ref{sec:transverse}).

The paper is organized as follows: In Section \ref{sec:adiabQA} we give a short introduction on adiabatic quantum annealing and define pertinent concepts that will be used later.
In Section \ref{sec:model} we define the adiabatic model subject to study and describe its main properties included the spectral gap.
In Section \ref{sec:speedup} we take advantage of the precise knowledge about the gap to prove the existence of a schedule that guarantees quadratic speedup over naive search for any problem instance.
Furthermore, we prove that quadratic speedup is optimal, meaning there is no schedule that achieves better time complexity than $\bigo{\sqrt{N}}$.
In Section \ref{sec:transverse} we support our claims with numerical simulations using a transverse field Hamiltonian as the initial Hamiltonian.
In Section \ref{sec:conclusions} we summarize our results.


\section{Adiabatic QA: preliminary definitions}\label{sec:adiabQA}

In a standard adiabatic QA algorithm one specifies a time-dependent Hamiltonian $H(z)$ that evolves according to the Sch\"odinger equation.
The parameter $z = z(\beta)$ is the {\it annealing schedule} and is a function of the dimensionless variable $\beta:=t/T$, where $t$ is the time and $T$ the total evolution time.
In the range $\beta\in[0, 1]$ the path $H(z)$ connects two Hamiltonians, $H_0 := H(z_i)$ and $H_f := H(z_f)$, where $z_i = z(0)$ and $z_f = z(1)$.
The initial {\it mixing} Hamiltonian $H_0$ has a ground state that is easy to prepare.
The final {\it target} Hamiltonian $H_f$ encodes in its ground state the solution to the optimization problem.
The goal of QA is to transform the ground state of $H_0$ to a state which is close to the ground state of $H_f$ at time $t=T$ by evolving the Schr\"odinger equation.

We say that the evolution is adiabatic if it is carried out slowly.
The error induced by this adiabatic evolution is given by the adiabatic theorem (AT)~\cite{Born1928, Kato1950}.
The story of the progress towards a mathematically clean AT is full of twists and turns.
The reader can consult~\cite{albash2018adiabatic} for a review.
In this work we use a rigorous version of the AT due to Jansen, Ruskai and Seiler~\cite{Jansen2007} as it appears in~\cite{albash2018adiabatic}.
The AT gives a bound on the adiabatic error in terms of the gap and derivatives of the Hamiltonian.
It states that for achieving an error less than $\epsilon$ in the comparison between the evolved state and the desired final ground state, the total evolution time $T$ must satisfy
\bal{\label{eq:AT}
T\ge \frac{C}{\epsilon} \,,\quad C := 2\max_\beta \frac{\|\dot H(\beta)\|}{g(\beta)^2} + \int_0^1 \left( \frac{\|\ddot H \|}{g^2} + \frac{7\|\dot H \|^2}{g^3} \right) d\beta\,,
}
where $g$ denotes the spectral gap of $H$ and where the dot represents derivative with respect to $\beta$. The value of $C$ is a measure of the time complexity of the model, but it doesn't work as a general definition of the complexity/cost of the algorithm since it is not scale-invariant
(one can make the complexity arbitrarily small by changing the time scale).
A proper measure of a scale-invariant cost is $T \max_z \| H(z)\|$~\cite{albash2018adiabatic}, where the maximum is taken in the interval $[z_i, z_f]$.
In cases where the maximum norm of the Hamiltonian is $\bigo{\text{poly}(n)}$ with $n$ the number of qubits, as it is for the present model, we can simply take $\text{cost} = T$ (we ignore non-exponential corrections in our complexity estimations).
Hereby, we take $T$ as a well-defined measure of the cost of the model.


\section{An adiabatic algorithm for combinatorial optimization}\label{sec:model}

\subsection{The model}
Consider an arbitrary optimization problem of the form
$$
s^* = \text{argmin}_s E(s)
$$
with $s=(s_1,\ldots,s_n)$ a vector of $n$ binary variables and $E(s)$ is an integer in the range $0\le E(s)\le m$ with $m=\text{poly}(n)$. Define $N_E$ the number of states with energy $E$ as $N_E=\sum_s \delta_{E(s),E}$, subject to $\sum_{E=0}^m N_E=N$ and $N=2^n$.
An example is 3-satisfiability (3-SAT) with $m\propto n$ clauses where $E(s)=\sum_{a=1}^m e_a(s)$ and each $e_a(s)=0,1$ depends on $3$ spins.
In this work we assume that $N_0 > 0$, i.e. we have at least one satisfying assignment.

We encode the optimization problem in an adiabatic model defining a spin Hamiltonian
\footnote{$H$ is related to the usual form  for quantum annealing $H = (1 - A) H_0 + A H_f$ with $0\le A \le 1$ by defining $z=\frac{A}{1 - A}$.}
\bal{\label{eq:model}
H = H_0 + z H_f\
}
where $H_f = \sum_s E(s)\ket{s}\bra{s}$ is diagonal in the $\sigma^z$-basis of $n$ spins,
and $H_0 = -\ket{\phi}\bra{\phi}$ with $\ket{\phi} = \frac{1}{\sqrt{N}}\sum_s \ket{s}$ the equal superposition state.
The parameter $z$ is the annealing parameter that changes with time $t$ varied from $z=0$ to its final value $z_f$. 
For $z=0$, $H$ has ground state $\ket{\phi}$. 
For large $z$, the ground state of $H$ encodes the minimal energy solution. 

The spectral analysis of model \cref{eq:model} can be performed with some precision. For details consult Appendix \ref{app:model}.
Here we outline the main points of the proof and refer the reader to Appendix \ref{app:model} for details.
The problem Hamiltonian $H(z)$ has a special form that allows to specify the first and second eigenvalues needed to compute the gap as the first two roots of a characteristic non-linear equation.
Due to a permutation symmetry present in $H(z)$ we can prove that the gap is given by $g = \lam_1 - \lam_0$, where $\lam_0$ and $\lam_1$ are the first two roots of the characteristic equation and satisfy $\lam_0 < 0 < \lam_1$ when $z>0$.
When the gap becomes small, each eigenvalue becomes small and close to zero.
We can then perform a Taylor approximation of the characteristic equation and compute $\lam_0$ and $\lam_1$ with high precision for large $n$.
Then, given an instance with density of states $\{n_E\}$ where $n_E := \frac{N_E}{N}$ and, by assumption, $N_0>0$, we obtain the following approximate expression for the spectral gap:
\bal{\label{eq:gap_approx}
g(z) =\frac{z}{Z_2}\sqrt{\left(z - Z_1 \right)^2 +4 \frac{N_0}{N} Z_2}
}
where $Z_p := \sum_{E=1}^m \frac{n_E}{E^p}$ is a partition sum encoding information about the problem instance.

From \cref{eq:AT} we see that the time complexity $T$ is roughly dictated by inverse powers of the gap and is dominated
when the gap reaches its minimum value.
Eq.~\cref{eq:gap_approx} implies that the minimal gap occurs at $z_* := Z_1 + \bigo{n_0}$ and is approximately given by
\bal{\label{eq:min_gap}
g_* := \min_{z}g(z)=2Z_1\sqrt{\frac{N_0}{NZ_2}}
}
to leading relative order in $\frac{1}{N}$. See Appendix \ref{app:model} for details.
From the expression of the minimal gap $g_*$ we see that to know its exact location we need detailed knowledge about the problem instance in the form of a partition sum $Z_1$ (see Section \ref{sec:fluctuations}).
This characterizes the complexity of model \cref{eq:model} in term of the spectral gap and its minimum value.
 Eq.~\cref{eq:min_gap} shows that the gap closes to a smallest value $g_*$ that scales exponentially in $n$, at a point $z_*$ which depends on detailed information (the density of states $\{n_E\}$) about the problem instance.
Then, the position of $g_*$ in the annealing interval fluctuates instance by instance.
This poses a problem at the moment of designing the schedule for achieving speedup as we will see in the next section.

To complete the definition of the algorithm we need a stopping condition.
In adiabatic QA this is simply the final evolution time $T$ where $z = z_f$.
The annealing point $z_f$ can be specified by requiring that at the end of the evolution the probability of finding the ground state solution is high enough,
which can be encoded through a free parameter in the model.
We find that a sufficient stopping condition is $z_f := z_*+\Delta$ with $\Delta := \sqrt{\frac{4N_0 Z_2}{N\delta }}$ with $\delta\ll 1$ an arbitrary constant independent of $n$ (see Appendix~\ref{appendix:whentostop}).
This guarantees a high probability of finding the ground state solution at the final time.
It is sufficient to stop just outside the range of minimal gap to achieve a high probability ratio of finding the solution.


\subsection{Fluctuations in the minimal gap location}\label{sec:fluctuations}

Eq.~\cref{eq:gap_approx} explicitly gives the minimal gap and its location in terms of $Z_1$ and $Z_2$, which depend on the density of states $N_E$ and thus on the problem instance.
Note that the range of $z$ where the gap is small ($\cO(1/\sqrt{N})$) is very narrow and of width $\cO(1/\sqrt{N})$. 
Instance by instance fluctuations in the density of states $\delta n_E$   induces fluctuations of $\bigo{n^{-5/2}}$ in the minimal gap location  $z_*=Z_1$ (see Fig.~\ref{file16b} in Appendix~\ref{sec:3sat}), while the depth of the gap stays of  $\bigo{1/\sqrt{N}}$.
For random instances $n_E$ drawn from arbitrary distributions, the fluctuations in the minimal gap location $z_*$ is thus very large compared to the width of the minimal gap.
We illustrate this in Fig. \ref{fig:gap_fluct}, where we computed the gap $g(z)$ dependency with $z$ for 20 random 3-SAT instances with $n=100$ spins. 
\begin{figure}
\begin{center}
\includegraphics[width=0.49\textwidth]{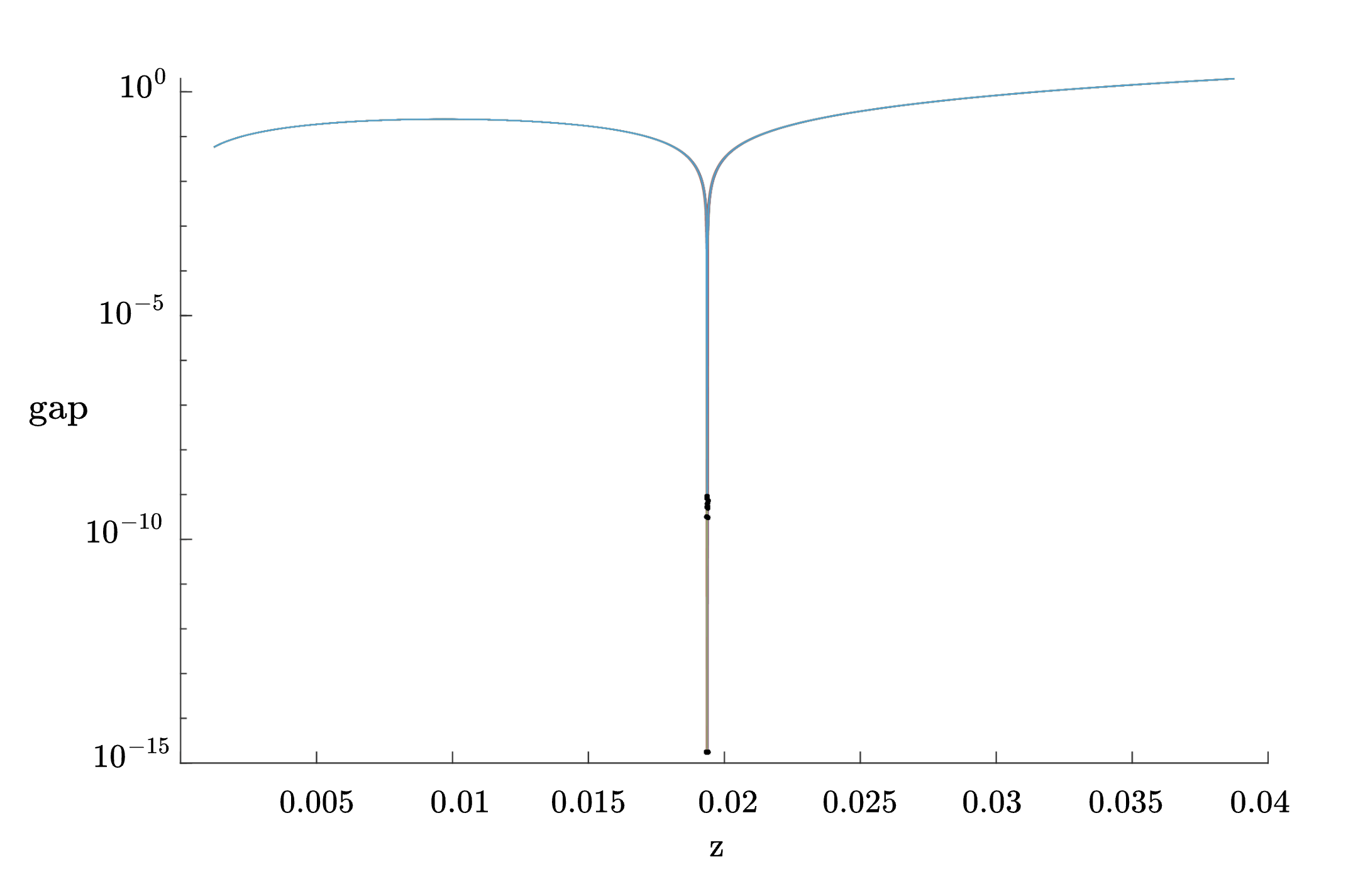}
\includegraphics[width=0.49\textwidth]{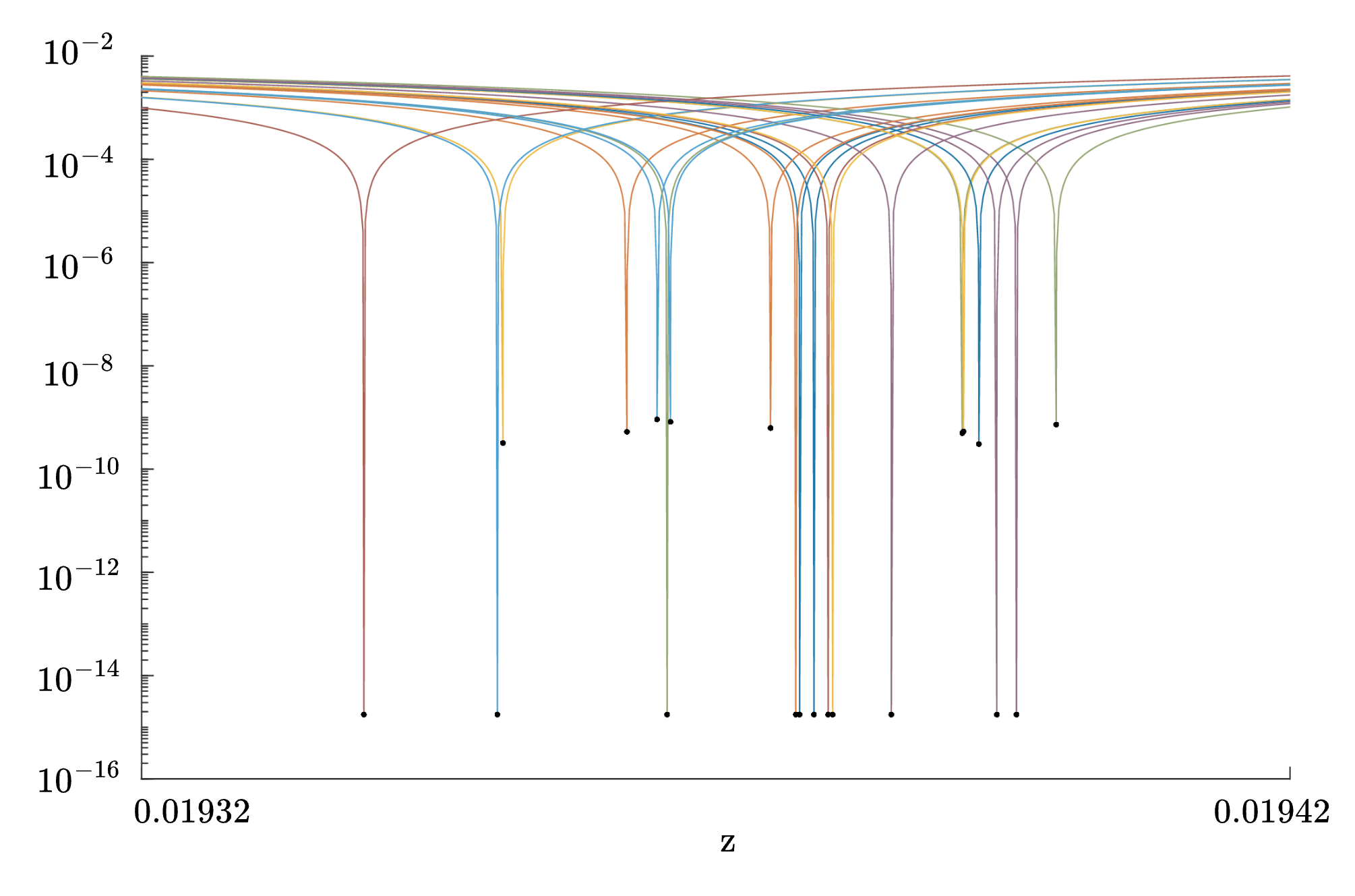}
\caption{{\bf Left}: Gap $g$ of~\cref{eq:gap_approx} versus $z$ for 20 random 3-SAT instances with density of states drawn from a Gaussian distribution $n_E\sim \cN(p_E,\Sigma)$ (see Appendix \ref{sec:3sat}) for $n=100$ spins.
The clause/spin ratio $\alpha := m/n=4.2$ close to  the
sat-unsat phase transition point~\cite{Mezard2002}.
{\bf Right}: Zoom-in on the minimal gap region.
The lower black dots with minimal gaps $\approx \num{e-15}$ correspond to instances with $N_0=1$ whereas the others with minimal gaps $\approx \num{e-9}$ correspond to $N_0>1$.
}
\label{fig:gap_fluct}
\end{center}
\end{figure}
The gap depends on the instance only through the density of states $n_E$.
For random 3-SAT, the mean density of states $\avv{n_E}$ is Binomial-distributed (see Eq.~\cref{binom} in Appendix \ref{sec:3sat}).
In the large $n$ limit, $n_E$ becomes Gaussian-distributed with covariance matrix given by Eq.~\cref{eq:nene} in Appendix \ref{sec:3sat}.

In order to obtain a quantum speedup, the annealing schedule (how $z$ varies with $\beta$) should be optimized such that annealing is fast when $g$ is large and (very) slow when $g$ is small. 
Due to the strong instance dependence of $g(z)$, the schedule is also instance dependent. This is possible, 
but requires the exact numerical value of $Z_1$, since that determines the location of the minimal gap. 
When $Z_1$ is known, or can be efficiently computed, one can construct an annealing schedule that computes the solution in $\bigo{\sqrt{N}}$ and thus achieves quadratic speed-up relative to a naive exhaustive search. An example is Grover search where $E(s)$ has only values 0, 1. The optimal schedule has $\dot{z}=\bigo{1/\sqrt{N}}$ in an exponentially small range of $\bigo{1/\sqrt{N}}$ around $z=Z_1$ and large $\dot{z}$ elsewhere.
For general optimization problems, however, $Z_1$ is intractable to compute \cite{MacKay2003} and an efficient algorithm cannot be designed.
The problem of computing the partition sum or, equivalently, the density of states, is a counting problem that is at least as hard as the corresponding NP problem~\cite{arora2009computational}.
An example is 3-SAT, which is NP-complete.
Computing the partition sum requires computing the density of states $\{n_E\}$.
In particular, it requires computing $n_0 = N_0/N$.
Computing the total number of solutions $N_0$ can be a harder task than computing the solutions themselves, since an efficient method to count solutions can be iteratively applied order $n$ times to find a particular solution in $\polyn$ runtime.

A possible way to avoid this negative conclusion, is to note that the cost function $0\le E(s)\le m$ can always be transformed to a new cost function $0\le E'(s)\le 1$ where $E'(s)=1$ iff $E(s)>0$. Then the density of state is known: $N_{E'=0}=N_0, N_{E'=1}=N-N_0$ in terms of the number of optimal solutions $N_0$. In terms of $E'$, the problem is equivalent to Grover's unstructured search, for which we already know that quadratic speed-up is possible \cite{roland2002quantum}.
But formulating the optimization problem in this way, yields a Hamiltonian $H_f$ that is no longer the sum of local terms, which may prevent an efficient circuit implementation \footnote{The Hamiltonian $H_0$ can be implemented on a quantum circuit by using an ancilla qubit with $\cO\left(\text{poly}(n)\right)$ controlled-NOT and single-qubit gates \cite{nielsen2010quantum}.}
Essentially, we have reduced the problem to Grover-like search at the expense of the locality of the target Hamiltonian.
As mentioned earlier, the problem of computing $N_0$ can be a challenging task for classical algorithms.
In the quantum circuit model, the original Grover search algorithm does not encounter the fine-tuning problem inherent to analog quantum computing.
However, it requires a precise knowledge of $N_0$ to determine when to stop the algorithm due to the souffl\'e problem~\cite{brassard1997}.
Gate-based algorithms exist to address this challenge and achieve quadratic speed-up for unknown $N_0$~\cite{Boyer1998, Chuang2014}.
In the continuous case, we can circumvent the need to compute the location of the minimal gap by recasting the problem in the form of the original adiabatic Grover algorithm introduced in~\cite{roland2002quantum}.
This algorithm features a fixed location of the minimal gap within the annealing interval and is independent of $n$.
Nevertheless, we still rely on knowing $N_0$ to design the schedule that achieves quadratic speed-up.
Even more, this algorithm already faces some precision requirements that can become computationally demanding with increasing number of qubits~\cite{albash2018adiabatic}.
Setting aside precision issues, it was proven in~\cite{Chuang2017} that the adiabatic Grover algorithm, together with a tailored schedule~\cite{roland2002quantum}, can achieve quadratic speed-up for estimates $N_0'$ such that $N_0' < N_0$, which translates into runtimes which are longer than necessary.


\subsection{Quadratic speedup for generic optimization problems}\label{sec:speedup}

From the AT \cref{eq:AT} we see that the time complexity of QA depends on the schedule $z(\beta)$.
An example is adiabatic Grover search where a naive schedule results in $T \sim N$ while an optimized version results in $T \sim \sqrt{N}$, the celebrated quadratic speedup~\cite{roland2002quantum}.
Here we aim to generalize this result to arbitrary optimization problems.
In this section we design a schedule that achieves quadratic speedup for the general case, from which Grover search is a particular case.
Later we prove that this speedup is optimal in the sense that this schedule minimizes $T$.

For the construction of the schedule we use a local-optimization ansatz that was first introduced in~\cite{roland2002quantum} (see also~\cite{Vazirani2001}) and further analyzed in \cite{Jansen2007}. Assume a schedule $z(\beta)$ that satisfies
\bal{\label{eq:ans}
\dot z = c g^2\
}
where $g = g(z)$ is the gap of model $H$ in \cref{eq:model} and $c$ is a normalization constant.
Note that Eq.~\cref{eq:ans} forces the schedule to slow down where the gap becomes small, which obviously requires knowing the position of its minimum.
Eq.~\cref{eq:ans} represents a boundary-value problem subject to conditions $z(0) = 0$ and $z(1) = z_f$.
Hence, $c$ is completely determined:
\bal{\label{eq:c}
c = \int_0^{z_f} \frac{dz}{g(z)^2}\,.
}
We use the AT stated in Eq.~\cref{eq:AT} to compute the time complexity of model $H$.
The complexity $C$ in \cref{eq:AT} can be bounded using the ansatz \cref{eq:ans}.
It can be shown that~\cite{albash2018adiabatic}
\bal{\label{eq:T}
C \le c \left( 2 E_{\text{max}} + 28 \int_0^{z_f} \frac{dz}{g(z)}\right)
}
with $c$ given by \cref{eq:c} and $E_{\text{max}} := \|H_f\| = \bigo{n}$.
Note that unlike \cref{eq:AT}, equation \cref{eq:T} depends on the schedule only through the gap,
meaning that we have gotten rid of schedule parametrization details and focus instead on the form of the gap along the adiabatic line $H(z)$.
We use this bound and the analytical approximation for the gap \cref{eq:gap_approx} to show the following.
\begin{lemma}
\label{lemma:quad_speedup}
Given the model $H = H_0 + z H_f$, there exist a schedule $z(\beta)$, solution to the boundary-value problem~\cref{eq:ans}, that achieves a time complexity
$$
T = \bigo{\sqrt{N/N_0}}\,.
$$
\end{lemma}
\begin{proof}
See Appendix \ref{app:existence_proof}.
\end{proof}

Lemma \ref{lemma:quad_speedup} states that for any optimization problem that can be encoded in $H$ there exists a schedule achieving quadratic speedup.
This was previously shown for Grover search \cite{roland2002quantum}, for which an analytical solution to \cref{eq:ans} is available.
Here we generalize this to the case of arbitrary target Hamiltonians that are diagonal in the computational basis.

Now we can ask whether this is optimal, or whether there exists a schedule that yields better than quadratic speedup, by optimizing the schedule using a different strategy than \cref{eq:ans}.
Previously, it was shown that quadratic speed-up is optimal for the Hamiltonian $H=(1-A)H_0+AH_f$ using a schedule $A(\beta)$ that is linear in $\beta$~\cite{farhi2008make}.
We generalize this result and prove that quadratic speedup is optimal for arbitrary schedules.
\begin{lemma}\label{lemma:quad_optimal}
Let $T$ the total evolution time of the adiabatic model $H = H_0 + zH_f$ and $P$ the projector onto the ground state subspace of $H_f$ with degeneracy $N_0$.
Define $p := \bra{\psi(T)} P \ket{\psi(T)} > 0$. Then
\bal{\label{eq:Tbound}
\sqrt{\frac{N}{N_0}} \gamma \le T\
}
with $\gamma := \frac{p}{2} \frac{\left( 1 - \sqrt{\frac{N_0}{Np}} \right)^2}{ 1 + \sqrt{p} }$ a constant of $\cO(1)$.
\end{lemma}

\begin{proof}
See Appendix \ref{sec:optimal_proof}.
\end{proof}

Note that Lemma~\ref{lemma:quad_optimal} states that 
the optimal time complexity scales as $T \sim \frac{1}{g_*}$, see (\ref{eq:min_gap}), while a naive 
estimate based on the adiabatic theorem predicts a scaling $T \sim \frac{1}{g_*^2}$~\cite{Vazirani2001}.
Applied to our case, this would yield $T \sim N/N_0$, i.e. no speedup at all.
Also, note that this lower bound on the annealing time applies regardless we have detailed knowledge of the problem instance or if the quantum annealing is performed non adiabatically.

The time complexity in \cref{eq:Tbound} depends on the number of solutions $N_0$. This number is constant and independent of $n$ for Grover search \cite{roland2002quantum} for which $N_0$ is fixed. In this case the complexity is $\cO(\sqrt{N})$.
The same complexity is obtained for the case of {\it unique satisfying assignments} in random 3-SAT~\cite{Hen2011} where one considers instances with only one solution.
For random 3-SAT, the number of solutions $N_0$ is a random variable that has mean value $\braket{N_0} = N p_0$, with $p_0 = \left(\frac{7}{8}\right)^{\alpha n}$ and $\alpha = m/n$ a constant (see Appendix \ref{sec:3sat}).
Then $\braket{N_0} \sim N^\gamma$ with $\gamma = 1 - \alpha |\log{\frac{7}{8}}|$.
Knowing that $g_*$ goes like $\sqrt{N_0/N}$ we can estimate the mean value of the minimal gap as
$
\braket{g_*} \sim N^{-\frac{\alpha}{2k}}
$
where $k := 1/|\log{\frac{7}{8}}| \approx 5.19$.
This value of $k$ was the first reported as an upper bound on the phase transition point $\alpha_c$ for random 3-SAT~\cite{Franco1983, Chvatal1988}.
Today, this value can be accurately estimated and is known to be around $\alpha_c = 4.26$ \cite{Mezard2002, braunstein2002survey, mertens2006threshold}.
For $\alpha < k$ the mean number of solutions $\braket{N_0}$ grows exponentially fast with $n$, while for $\alpha > k$ it decreases with the same rate.
Using the optimal schedule, the relation $T \sim 1/g_*$ implies $T \sim N^{\frac{\alpha}{2k}}$.
Thus, the speedup for the average instance with $N_0=\braket{N_0}$ depends monotonically on $\alpha$,
reaching a $\sqrt{N}$ complexity when $\alpha = k$, just where $\braket{N_0} = 1$, i.e. the case of unique satisfying assignments.
Although it is interesting to note that the speedup can be better than $\sqrt{N}$ for $\alpha < k$, it still implies an exponentially large computation time.
For smaller $\alpha$ there are also efficient classical algorithms to solve random 3-SAT instances~\cite{Mezard2002}.


\section{Comparison with the common mixing Hamiltonian}\label{sec:transverse}

One could argue that our negative result only holds for the rank one projector Hamiltonian $H_0$ that we used, and much optimistic results could be obtained using the common transverse field Hamiltonian $H_0 = -\sum_{i=1}^n \sigma_i^x$.
This is a valid concern. 
In this section we numerically investigate the spectral gaps up to $n=20$ qubits by using the transverse field Hamiltonian as our mixing Hamiltonian and producing the diagonal Hamiltonian ($H_f$) energies with two different models.
The first one is the 3-SAT model, and the second one belongs to a class of Ising spin glass models described by 3-spin interactions.
Numerically we see that in the worst case, for 3-SAT instances the minimal gap approximately scales like $\frac{1}{\sqrt{N}}$, similar to the case with rank one projector Hamiltonian, however for the spin glass instances, the minimal gap vanishes faster than $\frac{1}{\sqrt{N}}$, indicating a longer adiabatic annealing time for the common transverse field Hamiltonian compared to the case in which the rank one projector Hamiltonian is the mixing Hamiltonian. 
Early studies on small instances indicated that the median spectral gap in this case scales as $1/\text{poly}(n)$, suggesting exponential speedup \cite{young2008size, farhi2001quantum, hogg2003adiabatic, farhi2000quantum, schutzhold2006adiabatic}. However, later it was shown that for the worse case instances the gap closes exponentially with $n$ and that these instances dominate for large $n$ \cite{young2010first}. 
In agreement with these findings, we find that the majority of small instances  have large gaps, but the worse case instances have exponentially small gap. Since the worse case instances dominate at large $n$, we expect that the spectral gap does not scale better than $1/\sqrt{N}$ for large instances. Therefore, in order to reach quadratic speedup with this Hamiltonian, one also needs an optimized annealing schedule that requires precise knowledge of the location of the minimal gap. There exists no known analytical expression for this location, but there is no reason to assume that it would be less intractable than in the case of Eq.~\cref{eq:model}.
From this we conclude that adiabatic quantum annealing also cannot yield better than quadratic speed up with this mixing Hamiltonian.

\begin{figure}[t]
    \centering
    \begin{subfigure}[t]{0.45\textwidth}
         \includegraphics[width=\textwidth, height=5.5cm]{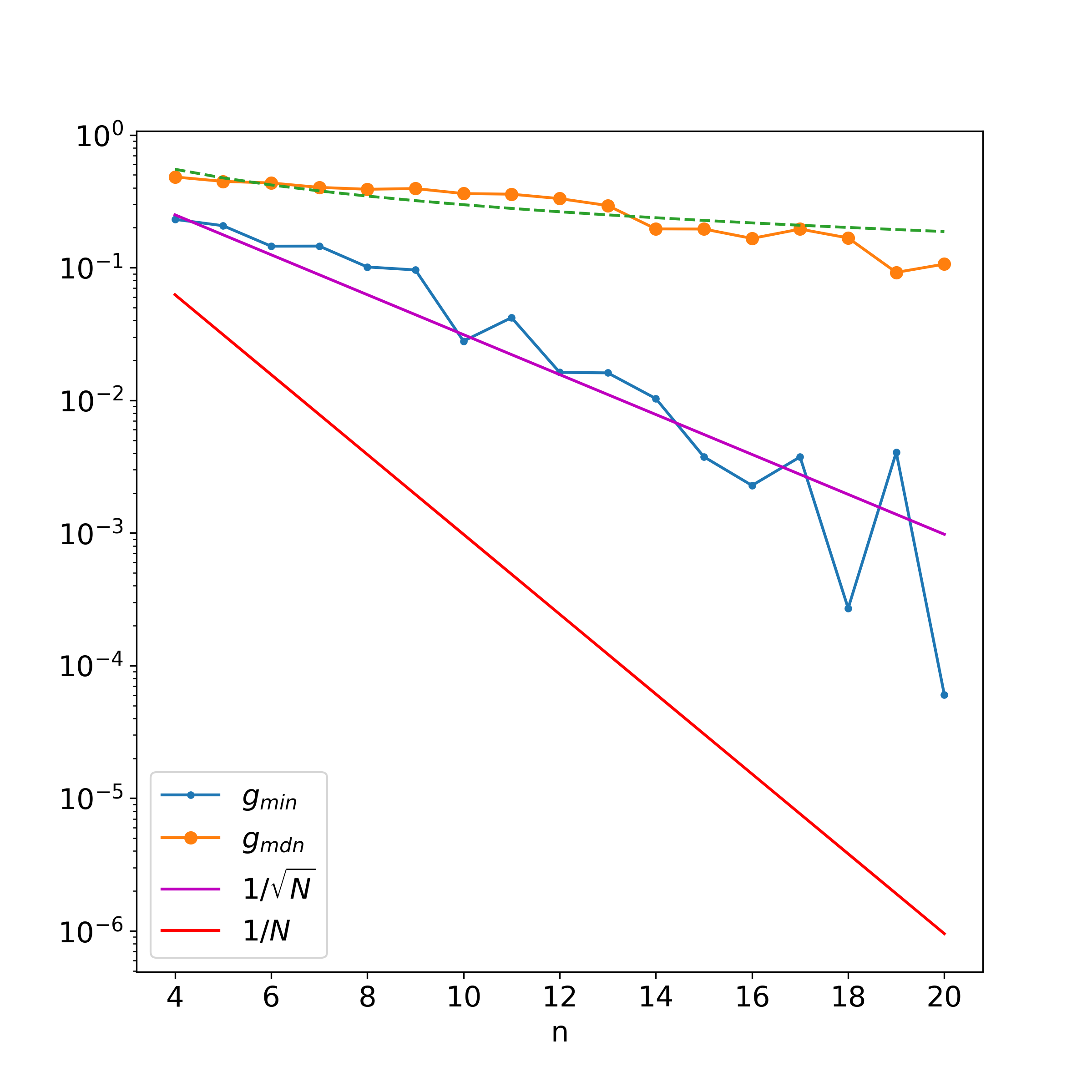}
     
     \end{subfigure}
     \begin{subfigure}[t]{0.45\textwidth}
         \includegraphics[width=\textwidth, height=5.5cm]{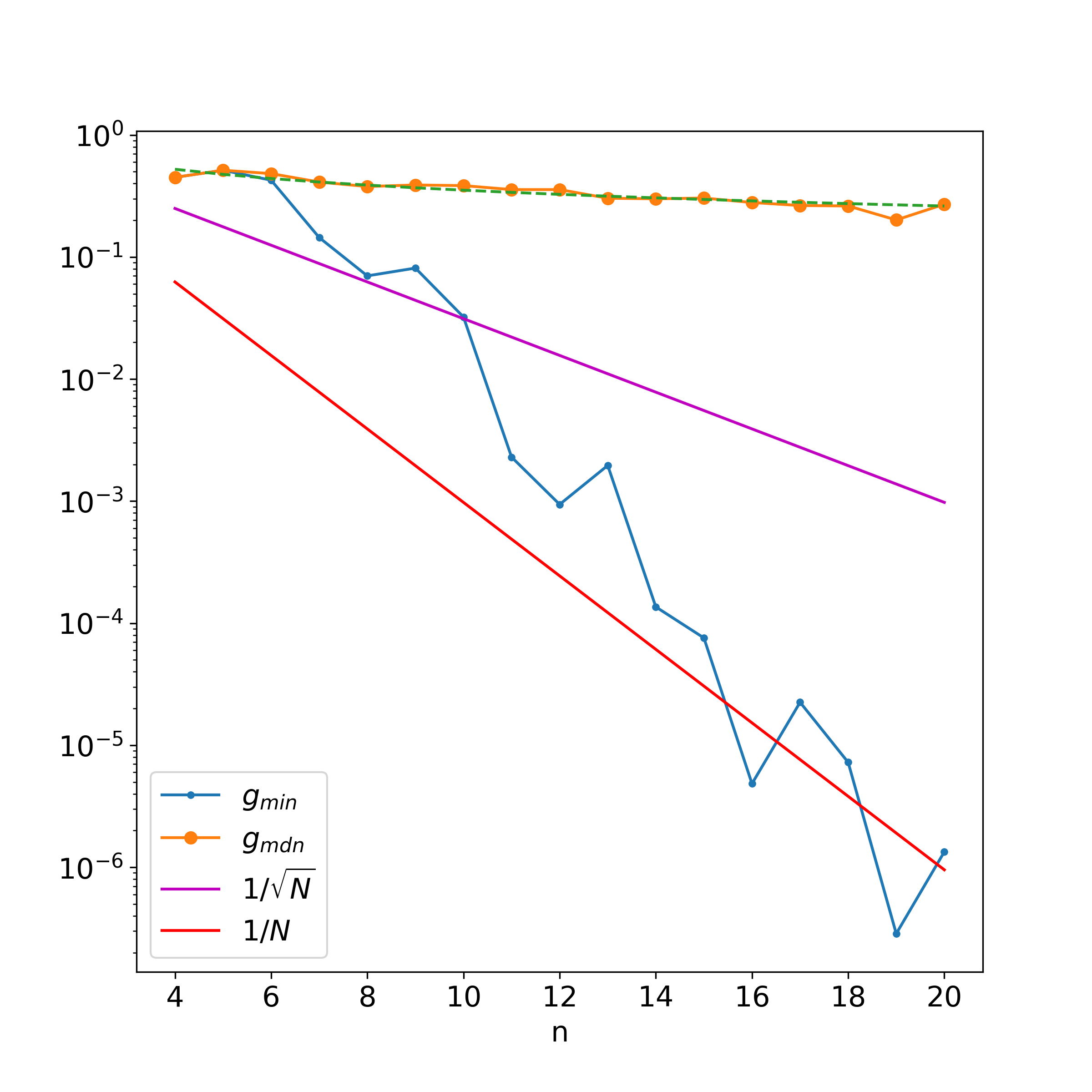}
         
     \end{subfigure}
    \caption{Scaling of the $g_*$ for $H = (1-A)H_0 + A H_f$, with $H_0 = -\sum_{i=1}^n \sigma_i^x$. {\bf Left} corresponds to 3-SAT instances with unique solutions ($N_0=1$). {\bf Right}, to 3-spin model instances with unique ground state.
    For each $n$, we computed the minimal gap $\min_A g(A)$ for a batch of different instances. We show the median and minimum gap over instances, denoted by $g_{mdn}$ and $g_{min}$ respectively, and the dashed line corresponds to a $1/\text{poly}(n)$ fit for $g_{mdn}$. We generated 200 instances for $n \le 13$ and 50 instances for $n>13$. The figure indicates that up to 20 qubits, minimal gap of the 3-SAT instances in the worst case scales approximately as $1/\sqrt{N}$ and for 3-spin model instances it indicates a close to  $1/N$ scaling as the number of qubits is increased.}
    \label{fig:scaling}
\end{figure}

Delving into details, we numerically studied the minimal gap of the Hamiltonian $H(A) = (1-A)H_0 + A H_f$ where the schedule $A\in[0, 1]$.
This setup is equivalent to model \cref{eq:model} as we stated in Sec.~\ref{sec:model}.
We use as the mixing Hamiltonian the transverse field $H_0 = -\sum_{i=1}^n \sigma_i^x$.
We consider target Hamiltonians $H_f$ associated to two different problems.
The first one corresponds to 3-SAT with $\alpha=m/n=4.26$.
The second, to an Ising spin glass Hamiltonian described by 3-spin interactions with couplings $J$ uniformly sampled from $\{\pm 1\}$:
\bal{\label{eq:3spin}
H_f = \sum_{i<j<k} J_{ijk} \sigma_i^z \sigma_j^z \sigma_k^z \,, \quad J_{ijk} = \pm 1 \,.
}
We refer to this model shortly as the 3-spin model.

For 3-SAT we consider instances with unique solutions ($N_0 = 1$).
For the 3-spin model we only consider instances with unique ground state.
At each $n$ and for each instance, we compute $g_*=\min_A g(A)$.

In Fig.~\ref{fig:scaling} we plot the median $g_{mdn}$ and the minimum $g_{min}$ of the minimal gap over instances as a function of $n$ for each model.
We see that in both cases and up to $n=20$ qubits, $g_{mdn}$ scales as $1/\text{poly}(n)$, indicating polynomial median complexity, whereas
$g_{min}$ decreases exponentially, indicating exponential complexity for the worst cases.
In the case of 3-SAT the gap decreases like $1/\sqrt{N}$ for worst cases.
This indicates that the complexity of  adiabatic QA for 3-SAT with the transverse field $H_0$ is similar to the projector $H_0$.
This supports our argument that the exponential shrinkage of the gap for hard instances will not disappear by using a different instance-independent mixing Hamiltonian.
However, with the transverse field $H_0$, the prospects of constructing a schedule that achieves the quadratic speedup are dim,  because we cannot compute $z_*$.
In simulations we observe that $g(z)$ may even have multiple local minima as a function of $z$.

For the 3-spin model, our simulations show a  scaling of the minimal gap close to $1/N$ which is worse than  projector $H_0$, where  the gap does not vanish faster than $1/\sqrt{N}$.
This results disagrees with previous studies~\cite{bapst2013quantum}, who reported exponentially small gaps of order $1/\sqrt{N}$ close to the phase transition point. In any case, the point is that the results with a transverse field $H_0$ are not better than with the projector $H_0$.
Therefore, it seems implausible that the use of the transverse field $H_0$ rather than a projector $H_0$ will improve the worse case performance of an adiabatic QA algorithm.


\section{Conclusions and discussion}\label{sec:conclusions}

In this paper we have addressed the question of whether quantum speed-up is possible with adiabatic quantum annealing. For this purpose, we have studied a class of Hamiltonians Eq.~\cref{eq:model} which allows us to compute the gap $g(z)$ around its minimal value $g_*=g(z_*)$.
The minimum location $z_*$ is given in terms of a partition sum that is in general intractable to compute.
Given $z_*$ we can construct an annealing schedule for which the time complexity scales as $\bigo{\sqrt{N}}$.
This was first shown for adiabatic Grover search in~\cite{roland2002quantum}, which is a particular instance of our model.
Consequently our construction generalizes this result to arbitrary target Hamiltonians that are diagonal in the computational basis.
Furthermore, we prove that for this Hamiltonian quadratic speed-up is optimal.
This extends earlier results from~\cite{farhi2008make} in which it was shown the same optimality for linear schedules, to the case of arbitrary non-linear schedules.
Thus, adiabatic quantum annealing with Hamiltonian given by~\cref{eq:model} gives at best quadratic speed-up, but only when $z_*$ is tractable to compute, which is not the general case.

The fact that quantum annealing does not yield exponential speedup should come as no surprise. The same holds in fact for classical Simulated Annealing. For instance, \cite{haj88} shows that in order to guarantee an optimal solution, the temperature should be decreased very slowly and the time of the algorithm scales exponential with $n$. This does not say anything about the potential practical value of Simulated Annealing. In fact, it is among the most powerful methods to find approximate solutions for combinatorial optimization problems. The same could be true for quantum annealing. 

A further issue that complicates adiabatic quantum annealing is the fact that, even when $z_*$ is known or can be computed, it needs to be specified to an exponential precision, which prevents practical implementation on a circuit \cite{slutskii2019analog}. The precision problem was also pointed out in \cite{hen2019quantum} which discusses a classical algorithm for unstructured search that realizes quadratic speedup at the expense of an increasing precision in the initial conditions with system size.

Recently, there has been a proposal of quantum adiabatic spectroscopy to empirically find the location $z_*$ where the spectral gap is minimal~\cite{Jordi2022}. It may be interesting to investigate whether this method can be used to estimate the partition functions that we used in this paper.

It is also worth mentioning that the quadratic speed-up obtained in this work corresponds to a model that does not encode the structure of the problem instance in the driver Hamiltonian $H_0$.
Driver Hamiltonians harnessing such structure might improve the speed-up over the more unstructured/unbiased $H_0 = -\ket{\phi}\bra{\phi}$.
How much structure one needs to encode in the model to gain considerable speed-ups is a question that remains open in the quantum computing community~\cite{aaronson2022structure}.

Although we have shown that the prospect of quantum speed-up using adiabatic QA is unlikely, this does not precludes the pursue of non-adiabatic QA or hybrid approaches to achieve speedup.
Analytical methods to explore the complexity of general non-adiabatic algorithms are not fully available and more research needs to be done to fill this breach.

In that line, it is an interesting question to consider what is the response of the model considered here if we use a hybrid protocol.
Would it be possible to find a general time evolution that leads us towards finding the ground state solution with better scaling or at least states that represent quasi-solutions?
We leave the answer to these questions for future work.


\section{Acknowledgments}
We thank the anonymous referees for their fruitful comments and suggestions on a previous draft of this work.
PN acknowledge support from 
the 'Quantum Inspire – the Dutch Quantum Computer in the Cloud' project (NWA.1292.19.194) of the NWA research program 'Research on Routes by Consortia (ORC)', which is funded by the Netherlands Organization for Scientific Research (NWO).
All authors made equal contribution to this work.


\appendix

\section{Model analysis}\label{app:model}
We complement the paper with analytical results related to model \cref{eq:model}
$$
H = H_0 + zH_f
$$
where $z$ is the annealing parameter, $H_f$ is diagonal in the computational basis and $H_0 = -\ket{\phi}\bra{\phi}$ is the rank-one Hamiltonian with $\ket{\phi}$ the equal superposition state.
We perform the exact computation of the spectral and gap properties of $H$.

\subsection{Energy spectrum and eigenstates}
Denote $\lambda_0 \le \ldots \le \lambda_{N-1}$ the eigenvalues of $H$. Sort the $m'\le m$ distinct energies for which $N_E>0$  as $E_1< E_2<\ldots < E_{m'}$\footnote{Note that these energies are different from the eigenvalues of $H$.}.
Remember that $N_E$ is the number of states with energy $E$ and $E_i$ can take values in the set $\{0, 1, \dots, m\}$.
\begin{lemma}
\label{lemma1} $H$ has a unique ground state eigenvalue $\lambda_0< zE_1$. 
In addition, $H$ has one non-degenerate eigenvalue in each interval $zE_i< \lambda < zE_{i+1}, i=1,\ldots,m'-1$ and $N_{E_i}-1$ eigenvalues $\lambda=zE_i, i=1,\ldots,m'$. 

\end{lemma}
\begin{proof}
Given an arbitrary real vector $v\in \mathbb{R}^N$ and an invertible matrix $M_0\in\mathbb{R}^{N\times N}$, the determinant of $M=M_0+vv^\top$ satisfies $\det(M)=(1+v^\top M_0^{-1} v)\det(M_0)$~\cite{Zhou2007}.
Thus, 
\bal{
\det(\lambda I - H) &= X(\lambda) \prod_{E=0}^m (\lambda-zE)^{N_E} \\
                    &= p(\lambda) \prod_{i=1}^{m'}(\lambda-zE_i)^{N_{E_i}-1}
}
with 
\bal{
X(\lambda)=1+\frac{1}{N}\sum_{E=0}^m\frac{N_E}{\lambda-zE}=1+\frac{1}{N}\sum_{i=1}^{m'} \frac{N_{E_i}}{\lambda-zE_i}=\frac{p(\lambda)}{\prod_{i=1}^{m'} (\lambda-zE_i)}\label{Xlambda}
}
$p(\lambda)$ is a polynomial of order $m'$ and has thus $m'$ solutions that are given by the zero crossings of $X(\lambda)$. 
$X(\lambda)$ is singular for $\lambda=zE_i$, $i=1,\ldots,m'$ and $\frac{dX}{d\lambda}<0$ for all $\lambda\ne zE_i$. Thus $X(\lambda)$ has a unique zero crossing in the intervals $\lambda < zE_1$ and $zE_i < \lambda < zE_{i+1}, i=1,\ldots,m'-1$. All these eigenvalues are non-degenerate.
The remaining $N-m'$ eigenvalues are solutions of $\prod_{i=1}^{m'} (\lambda-zE_i)^{N_{E_i}-1}=0$. For each $E_i$, there are $N_{E_i}-1$ eigenvalues $\lambda=zE_i$. 
\end{proof}

We can analytically compute the eigenvectors of $H$ for those eigenvalues that satisfy $X(\lambda) = 0$ (i.e. those that satisfy $zE_i<\lambda <zE_{i+1}$, not the eigenvalues $\lambda=zE_i$). In particular,
\begin{lemma}
\label{lemma:eigenvectors}
Let $\lam$ be the eigenvalues of $H$ that satisfy $X(\lam)=0$. The corresponding eigenvectors $\ket{v} = \sum_s v(s)\ket{s}$ have components
\bal{
v(s)\propto \frac{1}{zE(s)-\lam}\,,
}
where $E(s)$ defines the target Hamiltonian $H_f = \sum_s E(s)\ket{s}\bra{s}$.
\end{lemma}
\begin{proof}
The eigen-equation is
\bal{
(H_0 + zH_f)\ket{v} = \lam \ket{v}\,,
}
which implies
\bal{\label{eq:eigeneq}
\ket{v} = - (zH_f - \lam)^{-1} H_0 \ket{v} = (zH_f - \lam)^{-1}\braket{\phi | v} \ket{\phi}\,.
}
Note that the inverse in \cref{eq:eigeneq} is well defined.
Since $\ket{\phi} = \frac{1}{\sqrt{N}}\sum_s \ket{s}$ and $H_f$ is diagonal in the computational basis,
$$
\ket{v} \propto \sum_s (zE(s) - \lam)^{-1} \ket{s}\,.
$$
\end{proof}

We wish to guarantee that the adiabatic algorithm finds a solution $s$ with $E(s) = 0$ when such a solution exists.
We therefore restrict ourselves to problem instances for which $N_0 > 0$.

By definition, the spectral gap is the energy difference between the first excited state and the ground state.
In our model, this might pose an apparent dilemma.
For it, let's take a look at the pattern
$$
\lam_0 < 0 < \lam_1\,,
$$
where $\lam_0$, $\lam_1$ now represent the first and second roots of $X(\lam) = 0$.
The value $\lam = 0$ is an eigenvalue of $H$ only when $N_0 > 1$.
In that case it carries degeneracy $N_0 - 1$.
When $N_0 = 1$, $\lam = 0$ is not an eigenvalue and we can safely define the gap as $\lam_1 - \lam_0$.
However, when $N_0 > 1$, it is raised the question of whether we should take the gap as $0 - \lam_0$ or $\lam_1 - \lam_0$.
It can be proven that the $0$ eigenvalue corresponds to states that belong to the solution space.
This implies that the true energy gap, the one that characterizes the separation between the solution space and first excited non-optimal states, is given by $\lam_1 - \lam_0$.
Next, we give a proof of this making use of a permutation symmetry present in $H$.

\begin{lemma}\label{lemma:gap_def}
The spectral gap of $H$ is given by $g := \lam_1 - \lam_0$, where $\lam_0,\, \lam_1$ are the first and second roots of $X(\lam) = 0$.
\end{lemma}
\begin{proof}

The proof starts by noticing that the whole quantum dynamics can be restricted to a space of dimension $m' < N$.
To see this, define the normalized states
\bal{\label{eq:a}
\ket{a_i} := \frac{1}{\sqrt{N_{E_i}}} \sum_{s | E(s) = E_i} \ket{s}\,, \quad  i = 1, \ldots, m'\,.
}
Denote the span of these states as $\cH_a$.
Consider an eigenvalue $\lam$ of $H$ that is a root of $X(\lam) = 0$ with eigenvector $\ket{v}$.
For each energy $E$ for which $N_E>1$ define permutation operators $R_E$ that interchange two arbitrary states $s$ corresponding to the same $E = E(s)$ and leaves the rest unchanged. For energies $E$ for which $N_E=1$ define $R_E$ as the identity operator. By construction, we have $[R_E, H] = 0$ for all $E$.
Since $\lam$ is a non-degenerate eigenvalue, its eigenvector $\ket{v}$ is invariant under the action of $R_E$, i.e. $R_E\ket{v} = \ket{v}$ for all $E$.
So we conclude that $\ket{v}$ should be a linear combination of $\ket{a_i}$ states, therefore $\ket{v} \in \cH_a$.
Since we have $m'$ eigenvectors with eigenvalues $\lam$ that satisfy $X(\lam)=0$, these eigenvectors also span $\cH_a$.
Then, the Hamiltonian $H$ can be written as
$$
H = H_a + H_a^\perp\,,
$$
where $H_a$ is an operator that acts only on $\cH_a$ and $H_a^\perp$ acts on the orthogonal complement space $\cH_a^\perp$.

Finally, note that the initial state $\ket{\phi}$ can be written as
$$
\ket{\phi} = \sum_i \sqrt{n_{E_i}} \ket{a_i}\,,
$$
where $n_{E_i} := N_{E_i}/N$ for $i=1, \ldots, m'$, so it also belongs to $\cH_a$.
Since $[H_a, H_a^\perp] = 0$ the entire quantum evolution becomes restricted to $\cH_a$.
The initial state has no component in $H_a^\perp$ and the dynamics will not develop a component in $H_a^\perp$ at any later time.
Thus, effectively, the Hamiltonian is given by
$$
H=H_a = \sum_{i=0}^{m' - 1} \lam_i \ket{v_i}\bra{v_i}\,.
$$
From this we conclude that the gap is given by $\lam_1 - \lam_0$.

\end{proof}


\subsection{Gap}\label{sec:gap_analysis}
Consider a given instance with density of states $\{n_E\}$ with $n_E = \frac{N_E}{N}$ and $N_0>0$ and define 
\bal{
Z_p := \sum_{E=1}^m \frac{n_E}{E^p}\,.
} 
We show the following result.
\begin{lemma}
\label{lemma5}
For a given instance, the lowest two eigenvalues of $H$, $\lambda_0< 0 < \lambda_1$, are approximately given by
\bal{\label{lambda01}
\lambda'_{1,0}=z\frac{z-Z_1\pm \sqrt{(z-Z_1)^2+4 \frac{N_0}{N} Z_2}}{2Z_2}\,.
}
where prime denote approximate quantities. The approximation for $\lambda_0$ has relative error $|\delta\lambda_0/\lambda_0| := \left|\frac{\lam_0 - \lam'_0}{\lambda_0}\right| \le \delta$, for $z\ge z_0 := Z_1 - \frac{Z_2^2}{Z_3}\delta$.
The approximation for $\lambda_1$ has relative error $|\delta\lambda_1/\lambda_1| := \left|\frac{\lam_1 - \lam'_1}{\lambda_1}\right|\le \delta$ for $z\le z_1 := Z_1 +\frac{Z_2^2}{Z_3}\delta$.
The spectral gap is approximately given by
\bal{\label{eq:gap_approx_proof}
g' =\lambda'_1-\lambda'_0=\frac{z}{Z_2}\sqrt{\left(z - Z_1 \right)^2 +4 \frac{N_0}{N} Z_2}\,.
} 
The approximation for $g$ has relative error $|\delta g/ g| := \left|\frac{g - g'}{g}\right| \le \delta$ for $z_0\le z\le z_1$. 
The minimal spectral gap $g_*$ occurs for $z = z_* := Z_1 + \bigo{n_0}$ where $n_0 = N_0/N$ and is given by
\bal{\label{gap_min}
g_*=\min_{z}g(z)=2Z_1\sqrt{\frac{N_0}{NZ_2}} +\bigo{n_0^{3/2}}}\,.
\end{lemma}

\begin{proof}
Since $g=\lambda_1-\lambda_0$, we must solve for the two lowest eigenvalues $\lambda_0,\lambda_1$ of the characteristic equation. From Lemma~\ref{lemma1}, $\lambda_0<0<\lambda_1$. When the gap $g=\lambda_1-\lambda_0$ is small, both $\lambda_0$, $\lambda_1$ are close to $0$.
For $z > 0$ and to leading order in $\lam$,
\bal{\label{eq:char_expansion}
X(\lambda)&= 1+\frac{N_0}{N\lambda}+\frac{1}{N}\sum_{E\ge 1} \frac{N_E}{\lambda - zE}=\frac{1}{Nz\lambda}\left(N_0 z + b \lambda -a\lambda^2 - \frac{NZ_3}{z^2}\lambda^3 + \cO(\lambda^4)\right)\,.
}
with $b := N(z-Z_1)$ and $a := N\frac{Z_2}{z}$.
The solution to the quadratic approximation to $X(\lam) = 0$ is
\bal{\label{eq:lambdas}
\lambda'_{1,0}=\frac{b\pm \sqrt{D}}{2a}\,,\qquad D = b^2+4N_0NZ_2\,.
}
We easily verify that this solution satisfies $\lambda_0<0<\lambda_1$. 
The gap is $
g' =\sqrt{D}/a$ which yields Eq.~\cref{eq:gap_approx_proof}. 
It is minimized when $z = z_* = Z_1 + \bigo{n_0}$ which gives Eq.~\cref{gap_min}.

The error in the gap estimate is due to the errors in $\lambda'_0$ and $\lambda'_1$.
For $z_*- \Delta z \le z\le z_*+ \Delta z$ both errors are small for some sufficiently small $\Delta z > 0$ and the gap is accurately estimated by Eq.~\cref{eq:gap_approx_proof}.
Denote $\delta \lambda_i$ the error in $\lambda'_i$ for $i=0, 1$.
To estimate $\delta \lambda_i$ we solve the third order correction to the quadratic approximation to $X(\lam) = 0$ as follows.
Consider
$$
X'(\lam) - N\frac{Z_3}{z^2} \lam^3 = 0
$$
where $X'(\lam) = -a \lam^2 + b \lam + c$ is the quadratic approximation in \cref{eq:char_expansion} and $\lam = \lam' + \delta\lam$ with $X'(\lam') = 0$. By solving to first order in $\delta\lam$, we have
\baln{
\delta\lam = \frac{N\frac{Z_3}{z^2}}{-2a\lam' - 3 N\frac{Z_3}{z^2} \lam'^2 + b} \lam'^3
}
Using \cref{eq:lambdas} and ignoring the $\bigo{\lam'^2}$ term in the denominator we have
\baln{
\delta\lam = \frac{N\frac{Z_3}{z^2}}{-2a\lam' + b} \lam'^3=\mp \frac{NZ_3}{z^2 \sqrt{D}}\lam'^3\,.
}
Therefore,
\bal{\label{eq:rel_error}
\left| \frac{\delta\lam}{\lam'} \right| = \frac{NZ_3}{z^2 \sqrt{D}}\lam'^2\,.
}
We wish to find a validity interval $z_0 \le z \le z_1$ in which the relative error remains bounded as
\bal{\label{eq:delta_bound}
\left| \frac{\delta\lam}{\lam'} \right| \le \delta\,,
}
for $\lam = \lam_0,\,\lam_1$ and some $\delta \ll 1$.
It can be easily shown that the relative error $|\delta\lam / \lam'|$ increases with $z$ for $\lam'_1$ and decreases for $\lam'_0$.
Therefore the maximum relative error is achieved at the boundaries $z_0$ and $z_1$.
For $\lam_0$ the maximum is achieved at $z=z_0$ and for $\lam_1$, at $z=z_1$.
Then, it is enough to find a suitable $\epsilon > 0$ such that $z_0 = Z_1 - \epsilon$ and $z_1 = Z_1 + \epsilon$ in order to satisfy condition \cref{eq:delta_bound}.
Combining \cref{eq:rel_error} with condition \cref{eq:delta_bound} and assuming that $\epsilon \gg n_0$ we obtain
$$
\epsilon \le \frac{Z_2^2}{Z_3}\delta\,.
$$
Therefore, by defining $\epsilon := \frac{Z_2^2}{Z_3}\delta$ we ensure that \cref{eq:delta_bound} is satisfied in the interval $Z_1 - \epsilon \le z \le Z_1 + \epsilon$.
Finally, note that the bound in Eq.~\cref{eq:delta_bound} directly bounds the relative error in the gap $|\delta g / g' | \le \delta$
\footnote{We prove this by noticing that $\frac{\delta g}{g'} = \frac{\delta\lam'_1 - \delta\lam'_0}{g'}$. Then, $\left| \frac{\delta g}{g'} \right| \le \delta \frac{|\lam'_1| + |\lam'_0|}{g'} \le \delta$, where in the last step we used the formulas for $\lam'_0$, $\lam'_1$ and $g'$.}
where $\delta g = g - g'$.

It remains to determine the dependency of $\delta$ with respect to $n$.
For consistency, $\delta$ has to be sufficiently small to satisfy $z_0>0$.
This requires that $\frac{Z_1 Z_3}{Z_2^2} > \delta$, meaning that if the quotient $\frac{Z_1 Z_3}{Z_2^2}$ decreases with $n$, $\delta$ must decrease with $n$ as well.
Using the definition of $Z_p$ and the fact that $1 \le E \le m = \bigo{n}$, it can be shown that $\delta = \bigo{1/n^2}$ is a sufficient condition to ensure $\frac{Z_1 Z_3}{Z_2^2} > \delta$.
Nevertheless, we give a typicality argument to support the claim that $\delta \ll 1$ independent of $n$ is enough to maintain consistency.
Assume that $Z_p$ can be approximated by $1/\braket{E}^p$ for large $n$.
This is true, in particular, for 3-SAT, where $\braket{E} = m/8$ with $m = \alpha n$ and $\alpha$ a constant.
Then, for large $n$, the quantity $\frac{Z_2^2}{Z_3 Z_1} \rightarrow 1$.
Numerical results for 3-SAT show that the quotient $\frac{Z_1 Z_3}{Z_2^2}$ asymptotically approaches one from above for large $n$.
Therefore, $\frac{Z_1 Z_3}{Z_2^2}$ is order 1 in the typical case, and taking $\delta$ sufficiently small and independent of $n$ is enough to ensure the consistency of the approximation.
\end{proof}
In Lemma \ref{lemma5} we specified the annealing interval $z\in [z_0, z_1]$ for which both $\lambda_0$ and $\lambda_1$ (and hence the gap) are accurately estimated.
We compare the approximation Eq.~\cref{eq:gap_approx_proof} with the exact gap for a random 3-SAT instance with $n = 20$ spins in Fig.~\ref{fig:gap_approx}.
It confirms that the errors in $\lambda_{0,1}$ and $g$ are small in their range of validity.
\begin{figure}[ht!]
\begin{center}
\includegraphics[height=0.6\textwidth]{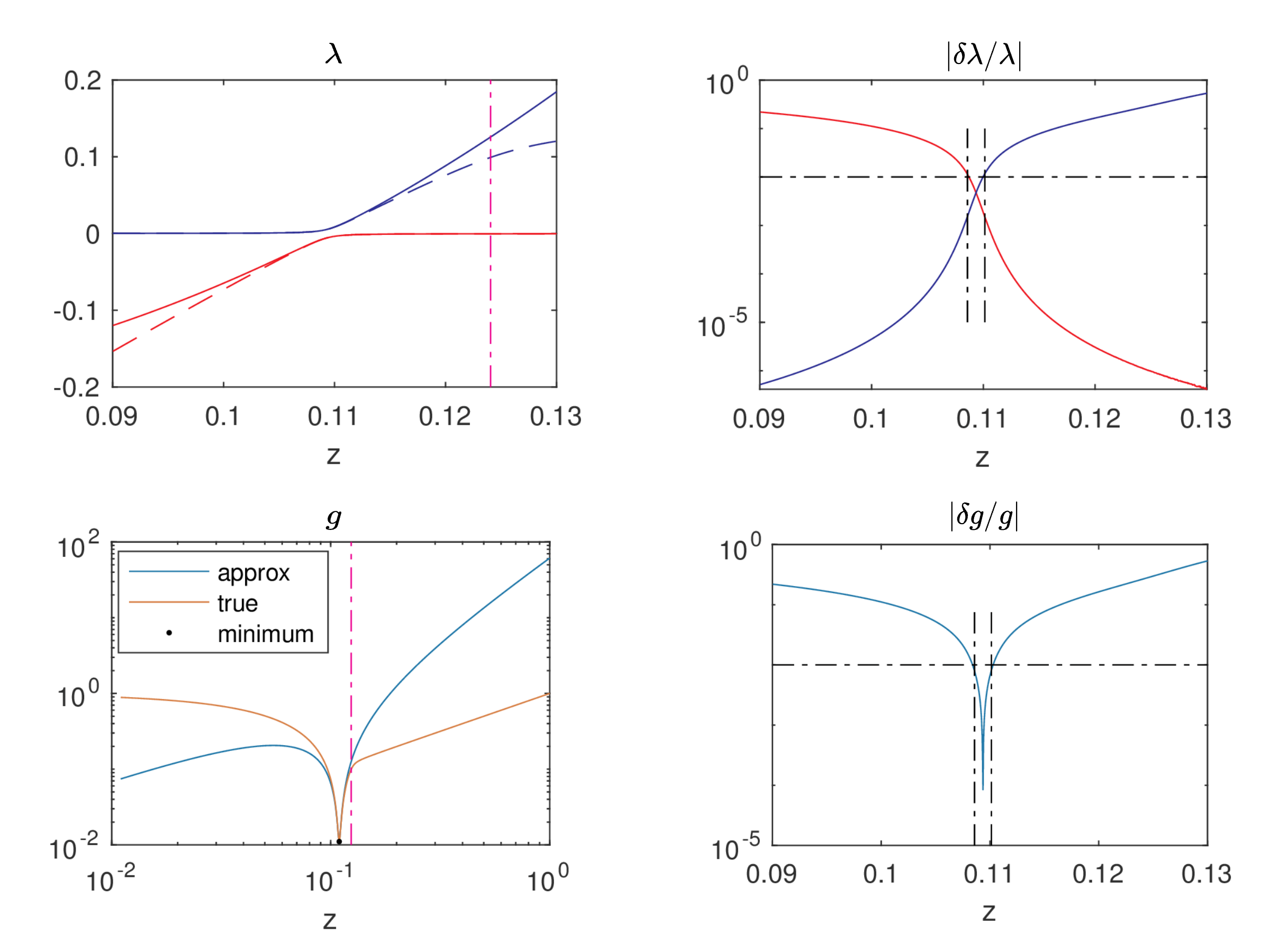}
\caption{
{\bf Left top}: Comparison between approximate (solid line) and true (dashed line) $\lambda_0$ (red) and $\lambda_1$ (blue) eigenvalues (see \cref{lambda01}) for one random 3-SAT instance with $n = 20$ and $\alpha = 4.2$. 
Dashed magenta line is detection bound $z_f=z_* + \sqrt{\frac{4 N_0Z_2}{N\delta}}$.
{\bf Left bottom}: Approximate vs true gap $g=\lambda_1 - \lambda_0$ (see \cref{eq:gap_approx_proof}) in the range $z\in [0, 1]$ and minimum location $(z_*,\, g_*)$ (black dot). 
{\bf Top right}: Relative error in the eigenvalues $\left| \delta \lambda_{0}/\lambda_{0}\right|$ (red) and $\left| \delta \lambda_{1}/\lambda_{1} \right|$ (blue).
{\bf Bottom right}: Relative error in the gap $\left| \delta g/g \right|$. Vertical black dashed lines indicate the validity intervals $z\ge z_0$ for $\lambda_0$ and $z\le z_1$ for $\lambda_1$ with $z_{1,0}=Z_1 \pm \delta \frac{Z_2^2}{Z_3}$ where the relative errors  $\left| \delta \lambda_{0,1}/\lambda_{0,1}\right|\le \delta=0.01$, respectively.
}
\label{fig:gap_approx}
\end{center}
\end{figure}

The next lemma assesses the behavior of the gap in the remaining of the annealing interval and it shows that it is not exponential.
\begin{lemma}
\label{lemma:large_gap}
For $0\le z\le z_0$ the spectral gap $g > \bigo{1/\polyn}$.
\end{lemma}

\begin{proof}
Define $\lambda=z\lambda'$. Then the characteristic equation $X(\lambda)=0$ becomes $F(\lambda')=\sum_{E=1}^m \frac{n_E}{\lambda'-E}=-z$. $F$ is a decreasing function of $\lambda'$. Increasing $z$, decreases the intercept $F(\lambda')=-z$ and thus increases $\lambda'$. Thus all solutions $\lambda$ of $X(\lambda)=0$ are increasing with $z$. 
The gap is $g(z)=\lambda_1(z)-\lambda_0(z)> -\lambda_0(z)$. 
Since $\lambda_0$ is increasing function of $z$ we have for $0\le z\le z_0$ that $\lambda_0(z)< \lambda_0(z_0)=-\frac{z_0Z_2}{Z_3}\Delta z$. Thus $g(z)> \frac{z_0Z_2}{Z_3}\Delta z =\bigo{1/\polyn}$. 
\end{proof}

This completes the characterization of the gap behavior along the annealing interval until the location of the minimal gap.
It shows the existence of a unique phase transition point at $z_*$ where the minimal gap $g_* \sim \sqrt{N_0/N}$ whereas in the rest of the interval the gap is non-exponential in $n$.

\subsection{When to stop the adiabatic evolution?}
\label{appendix:whentostop}
We compute the annealing point $z_f$ for which it is sufficient to achieve a high probability ratio of finding a solution at the end of the evolution.
\begin{lemma}
\label{lemma:prob}
Define $z_f := z_*+\Delta$ with $\Delta := \sqrt{\frac{4N_0 Z_2}{N\delta}}$ and $\delta\ll 1$ a constant independent of $n$.
The probability ratio in the ground state to detect the optimal solution $E=0$ 
\bal{
\frac{p(E=0)}{p(E>0)}\ge \frac{4 Z_2}{\delta}\,.
}
\end{lemma}
\begin{proof}
It is reasonable to assume that $z_f > z_*$, that is, stopping the evolution after the phase transition point $z_*$.
Then, we can use Eq.~\cref{lambda01}, which is accurate up to $\bigo{N_0/N}$ corrections, to derive $|\lambda_0|\le \frac{zN_0}{N\Delta}$ for $z\ge z_f$ , where $\Delta = \sqrt{\frac{4N_0 Z_2}{N\delta}}$.

Using Lemma~\ref{lemma:eigenvectors} the probability amplitudes are
\bal{
p(E=0) &\propto \frac{N_0}{\lambda_0^2} \ge \frac{N^2\Delta^2}{ N_0z^2} \\
p(E>0) &\propto \sum_{E\ge 1}^m \frac{N_E}{(zE-\lambda_0)^2}\le \frac{N}{z^2}\\
\frac{p(E=0)}{p(E>0)} & \ge \frac{N\Delta^2}{N_0}= \frac{4 Z_2}{\delta}\,.
}
\end{proof}

The final point $z_f$ is indicated in Fig.~\ref{fig:gap_approx} as the magenta dashed line.
Note that for large $n$ we have $z_f< z_1$ so that the annealing terminates within the range of $z$ values where the approximation of the gap is accurate. 


\section{Quadratic speed-up: proof of existence}
\label{app:existence_proof}

\begin{proof}
Here we provide a proof of Lemma~\ref{lemma:quad_speedup}. The integrals
$$
\int_0^{z_f} \frac{dz}{g(z)^p}\,,\quad p = 1, 2
$$
can be bounded as follows.
From Lemma~\ref{lemma5}, in the interval $[z_0, z_1]$ the gap is accurately approximated by \cref{eq:gap_approx}.
Set $z_f = z_* + \Delta < z_1$ with $\Delta = \sqrt{\frac{d}{\delta}}$, $d := 4n_0Z_2$ and $\delta \ll 1$ independent of $n$.
Therefore in $[z_0, z_f]$ and up to terms of the same order as the integral,
$$
\int_{z_0}^{z_f} \frac{dz}{g(z)^p} = \int_{z_0}^{z_f} \frac{Z_1^p}{z^p}
     \frac{dz}{\left(\sqrt{(z - Z_1)^2 + 4n_0 Z_2}\right)^p}
\leq \frac{Z_1^p}{z_0^p} \int_{z_0}^{z_f} \frac{dz}{\left( \sqrt{(z - Z_1)^2 + 4n_0 Z_2}\right)^p}\,.
$$
Hence, for $p=1$
\footnote{
For $p=1$ and $p=2$ we have the following primitives
$$
\int \frac{dx}{\sqrt{x^2 + d}} = \frac{1}{2} \ln{\frac{\sqrt{x^2 + d} + x}{\sqrt{x^2 + d} - x}} + \text{ct.}\,,
$$
and
$$
\int \frac{dx}{x^2 + d} = \frac{1}{\sqrt{d}} \tan^{-1}\left(\frac{x}{\sqrt{d}}\right) + \text{ct.}\,,
$$
respectively.
}
\baln{
\int_{z_0}^{z_f} \frac{dz}{\sqrt{(z - Z_1)^2 + 4n_0 Z_2}} &= \frac{1}{2}\left. \ln\left(
\frac{\sqrt{x^2 + d} + x}{\sqrt{x^2 + d} - x}
\right)
\right|^{\Delta}_{z_0-Z_1} \\
&=  \frac{1}{2}\ln\left(\frac{\sqrt{1 + \delta} + 1}{\sqrt{1 + \delta} - 1}\right) - 
\frac{1}{2}\ln\left(\frac{\sqrt{(z_0 - Z_1)^2 + d} + z_0 - Z_1}{\sqrt{(z_0 - Z_1)^2 + d} - (z_0 - Z_1)}\right)\,.
}

Note that $|z_0-Z_1|=\delta \frac{Z_2^2}{Z_3} \gg \sqrt{d}=2\sqrt{Z_2} \sqrt{\frac{N_0}{N}}$ for large $n$ when the number of optimal solutions $N_0$ is $\cO(N^\gamma)$  with $\gamma < 1$. Therefore, to leading order in $d$
$$
\ln\left(
\frac{\sqrt{(z_0 - Z_1)^2 + d} + z_0 - Z_1}{\sqrt{(z_0 - Z_1)^2 + d} - (z_0 - Z_1)}
\right) =- \ln\left(\frac{d}{4(z_0-Z_1)^2}\right) = \bigo{n}\,.
$$
Thus
$$
\int_{z_0}^{z_f} \frac{dz}{\sqrt{(z - Z_1)^2 + 4n_0 Z_2}} = \bigo{n}\,.
$$

For $p=2$
\baln{
\int_{z_0}^{z_f} \frac{dz}{(z - Z_1)^2 + 4n_0 Z_2}& = \frac{1}{\sqrt{d}} \tan^{-1}\left(\frac{x}{\sqrt{d}}\right)\bigg|^\Delta_{z_0 - Z_1} \\
 &= \frac{1}{\sqrt{d}} \left(
\tan^{-1}(1/\sqrt{\delta}) - \tan^{-1}\left(\frac{z_0 - Z_1}{\sqrt{d}}\right)\right)\\
 & \approx  \frac{1}{\sqrt{d}} \pi = \bigo{\sqrt{N/N_0}}\,.
}

Putting all together in \cref{eq:T},
\baln{
C &\le \left(\int_0^{z_0} \frac{dz}{g(z)^2} + \int_{z_0}^{z_f} \frac{dz}{g(z)^2} \right) \left( 2 E_{max} + 28 \int_0^{z_0} \frac{dz}{g(z)} + 28 \int_{z_0}^{z_f} \frac{dz}{g(z)} \right)\\
 &= \left(\int_0^{z_0} \frac{dz}{g(z)^2} + \bigo{\sqrt{N/N_0}} \right) \left( 2 m + 28 \int_0^{z_0} \frac{dz}{g(z)} + \bigo{n} \right)\,.
}
We know from Lemma \ref{lemma:large_gap} that $g > 1/\polyn$ in the interval $[0, z_0]$.
Since $z_0 \le 1$, then
$$
\int_0^{z_0} \frac{dz}{g(z)^p} < \bigo{\polyn}\,.
$$
Thus, $C = \cO(\sqrt{N/N_0})$ up to factors polynomial in $n$. Therefore the implicit schedule implemented by Eq.~\cref{eq:ans} achieves a time complexity
$$
T = \bigo{\sqrt{N/N_0}}\,.
$$
\end{proof}


\section{Quadratic speed-up: proof of optimality}\label{sec:optimal_proof}

\begin{proof}
For the proof of Lemma~\ref{lemma:quad_optimal} we take advantage of a permutation symmetry present in model~\cref{eq:model}.
Define the set of solution states $S=\{s|E(s)=0\}$. The Hilbert space $\cH_S := \text{Span}\{ \ket{s} | s\in S\}$ is the associated solution space.
Define the permutation operator $R$ which acts in the solution space $\cH_S$ by interchanging two solution state vectors.
Since $[R, H] = 0$, if we start the evolution with the state $\ket{\phi}$, which is invariant under the action of $R$, the symmetry is preserved in the solution space along the entire evolution.
Therefore, $\ket{\psi_t} := U(t)\ket{\phi}$, with $U(t)$ the propagator for $H$, is invariant under permutation in $\cH_S$.
Since the entire dynamics of the problem is restricted to $\cH_a$ (see proof of Lemma~\ref{lemma:gap_def}), we can write $\ket{\psi_t}$ as
\bal{\label{eq:symm}
\ket{\psi_t} = \sum_{i=0}^{m' -1} c_i(t)\ket{a_i} \,,\quad c(t)\in\complex\,.
}
Define the projector $P = \sum_{s\in S} \ket{s}\bra{s}$ onto the solution space $\cH_S$
\footnote{Note that $P$ is different from the projector onto the particular state $\ket{a_0} = \frac{1}{\sqrt{N_0}}\sum_{s\in S} \ket{s}$ (see Eq.~\cref{eq:a}).}.
We can relate $c_0(t)$ in \cref{eq:symm} with the detection probability $p$ at $t = T$ as
$$
p = \bra{\psi_T} P \ket{\psi_T} = |c_0(T)|^2 \rightarrow c_0(T) = \sqrt{p} e^{i\theta_T}
$$
for some $\theta_T$ real.

Define the {\it unnormalized} state
$$
\ket{\omega_t} := \frac{1}{\sqrt{p}} P \ket{\psi_t} = \frac{c_0(t)}{\sqrt{p}} \ket{a_0}\,.
$$
Note that $\ket{\omega_T}$ is normalized.
Define the error measure
\bal{\label{eq:error1}
E(t) := \left\| \ket{\omega_T} - \ket{\omega_t} \right\|^2\ = \left\| e^{i\theta_T}\ket{a_0} - \ket{\omega_t} \right\|^2
= 1 + \braket{\omega_t | \omega_t} - 2\,\re\, e^{-i\theta_T}\braket{a_0 | \omega_t}\,.
}
We differentiate the error \cref{eq:error1} wrt $t$ and obtain
\bal{\label{eq:deriv1}
\partial_t E = \frac{i}{p} \bra{\psi_t} [H, P] \ket{\psi_t} - \frac{2}{\sqrt{p}}\,\im\, e^{-i\theta_T} \bra{a_0} PH \ket{\psi_t}\,.
}
where we use that $\ket{\psi_t}$ satisfies the Schr\"odinger equation $i\partial_t \ket{\psi_t} = H \ket{\psi_t}$.
For $s\in S$,
$$
H\ket{s} = \left(z\sum_{s'}E(s')\ket{s'}\bra{s'}-\ket{\phi}\bra{\phi}\right) \ket{s} = -\frac{1}{\sqrt{N}}\ket{\phi}\,.
$$
The commutator can be computed from
\baln{
HP = \sum_{s\in S} H\ket{s}\bra{s} = -\sqrt{\frac{N_0}{N}} \ket{\phi}\bra{a_0}
}
as
\baln{
[H, P] = HP - (HP)^\dagger = \sqrt{\frac{N_0}{N}} \left( \ket{a_0}\bra{\phi} - \ket{\phi}\bra{a_0} \right)\,.
}
Inserting the commutator in \cref{eq:deriv1} and taking the absolute value we get
$$
|\partial_t E| \leq 2 \sqrt{\frac{N_0}{N}} \left( \frac{1}{p}  + \frac{1}{\sqrt{p}}\right)
$$
where we used that $P\ket{a_0} = \ket{a_0}$.
By integrating in time and using $|E(T) - E(0)| \le \int_0^T |\partial_t E| dt$,
$$
|E(T) - E(0)| \leq 2 T \sqrt{\frac{N_0}{N}} \left( \frac{1}{p}  + \frac{1}{\sqrt{p}}\right)\,.
$$
From definition \cref{eq:error1} we have
\bal{\label{eq:boundaries1}
E(T) = 0\,,\quad E(0) = \left| e^{i\theta_T} - \sqrt{\frac{N_0}{Np}} \right|^2\,.
}
Therefore
$$
\left| e^{i\theta_T} - \sqrt{\frac{N_0}{Np}} \right|^2 \leq 2 T \sqrt{\frac{N_0}{N}} \left( \frac{1}{p}  + \frac{1}{\sqrt{p}}\right)
$$
from which we obtain the final result
$$
\sqrt{\frac{N}{N_0}} \frac{p}{2} \frac{\left( 1 - \sqrt{\frac{N_0}{Np}} \right)^2}{ 1 + \sqrt{p} } \le T\,.
$$
\end{proof}


\section{Random 3-SAT: density of states approximation}\label{sec:3sat}

For 3-SAT on $n$ spins with $m$ clauses define $E(s):= E_s = \sum_{a=1}^m e_a(s)$ with
$$
e_a(s)=\prod_{i\in \cS_a}\frac{1+J_{ai}s_i}{2}
$$
and $\cS_a$ a subset of 3 spins. 
A specific instance of 3-SAT defines a density of states $N_E=\sum_{s=1}^N \delta_{E,E_s}$ with $N=2^n$.

In general, $N_E$ is intractable to compute. However, for random 3-SAT we can compute its expected value and the fluctuations, as follows. $J_{ai} = \pm 1$ are independent uniform random variables. Therefore, 
the $e_a(s)=0,1$ are independent binary variables with $p(e_a=1)=\frac{1}{8}$ and 
the probability that $E(s) = E$ is binomial distributed 
\bal{
p_E = \binom{m}{E}p^E(1-p)^{m-E}\,,\qquad p=\frac{1}{8}\label{binom}\,.
}
Therefore the expected density of states satisfies $\braket{n_E} = p_E$ where $n_E = \frac{N_E}{N}$.

To compute the variance in $n_E$ due to instance by instance fluctuations, we need to compute the correlation matrix $\Sigma_{EE'} := \braket{n_E n_{E'}} - p_E p_{E'}$.
\begin{lemma}
\label{lemma2}
For random 3-SAT 
\bal{\label{eq:nene}
\braket{n_E n_{E'}} = \frac{p_E}{N}\left(\delta_{E,E'}+\sum_{d=1}^{n}\binom{n}{d} p(E'|E,d)\right)
}
with $p(E'|E,d)$ given by Eq.~\cref{pcond} with $p_0,\,p_1$ depending on $d$ as given by Eqs.~\cref{eq:p11} and~\cref{p0p1}.
\end{lemma}
\begin{proof}
Since $N_E=\sum_s \delta_{E_s,E}$ we find 
\bal{\label{NENE}
\braket{N_E N_{E'}} & = \sum_{s,s'} \braket{\delta_{E_s,E}\delta_{E_{s'},E'}}=\delta_{E,E'} \sum_s \braket{\delta_{E_s,E}} +\sum_{s,s'\ne s} \braket{\delta_{E_s,E}\delta_{E_{s'},E'}}\nonumber\\
& = p_EN \delta_{E,E'} +\sum_{s,s'\ne s} \braket{\delta_{E_s,E}\delta_{E_{s'},E'}}\,.
}
In order to evaluate the second term, we need to compute the joint probability $p(E(s)=E,\,E(s')=E')$. 
For given $s\ne s'$, 
\bal{
e_a(s)e_a(s')=\prod_{i\in \cS_a}\frac{1}{4}(1+J_{ai}s_i)(1+J_{ai}s_i')\,.\nonumber
}
$e_a(s)e_a(s')=1$ iff $s_i=s'_i$ for all $i\in \cS_a$, ie. when $s_a=s_a'$, with $s_a, s_a'$ the three bits of state $s,s'$ in clause $a$.
Thus, by randomizing over the $J_{ai}$,
\bal{
p(e_a(s)e_a(s')=1|s_a=s_a')
=\text{Prob}
\left(\prod_{i\in \cS_a}\frac{1}{2}(1+J_{ai}s_i) = 1\right)
=\frac{1}{8}\,.\nonumber
}
and zero for $s_a\ne s_a'$. 
Denote $d$ the number of bits that $s, s'$ differ. The probability that $s,s'$ are  identical on the three bits in $\cS_a$ is 
\bal{\label{eq:p11}
p_1:=p(s_a=s_a'|d)=\frac{\binom{n-d}{3}}{\binom{n}{3}}\qquad d=1,\ldots n\,.
}
Note that $p(s_a=s_a'|d)=0$ when $d\ge n-3$\,. 
Thus
\bal{\label{p0p1}
p(e_a(s)=1,e_a(s')=1|d)&=p(e_a(s)=e_a(s')=1|s_a=s_a')p(s_a=s_a'|d)=\frac{1}{8}p_1\nonumber\\
p(e_a(s)=0,e_a(s')=1|d)&=p(e_a(s')=1|d)-p(e_a(s)=1,e_a(s')=1|d)=\frac{1}{8}(1-p_1)\nonumber\\ p(e_a(s')=1|e_a(s)=1,d)&=\frac{p(e_a(s)=1,e_a(s')=1|d)}{p(e_a(s)=1)}=p_1\nonumber\\ p(e_a(s')=1|e_a(s)=0,d)&=\frac{p(e_a(s)=0,e_a(s')=1|d)}{p(e_a(s)=0)}=\frac{1}{7}(1-p_1)=p_0\,.
}
We write $p(E(s)=E,E(s')=E'|d)=p(E)p(E'|E,d)$ with $p(E)$ given by Eq.~\cref{binom}.  We condition the probability of the outcomes of $e_a(s')$ on the outcomes of $e_a(s)$. 
We sort the outcomes of $E_{a=1:m}(s)$  as 
\bal{
e_a(s)  &= \underbrace{1,\ldots\ldots\ldots\ldots,1}_{E}\underbrace{0,\ldots\ldots\ldots\ldots 0}_{m-E} \nonumber\\
e_a(s') &= \underbrace{1,\ldots,1}_{x}\underbrace{0,\ldots,0}_{E-x}\underbrace{1,\ldots,1}_{E'-x} \underbrace{0,\ldots,0}_{m-E-(E'-x)}\,.\nonumber
}
For $a=1:E$, $e_a(s')=1$ $x$ times and zero $E-x$ times. For $a=E(s)+1:m$, $e_a(s')=1$ $E'-x$ times and zero $m-E-(E'-x)$ times. $x$ can range as $0\le x\le E$ and $0\le E'-x\le m-E$.  The latter expression is equivalent to $E+E'-m\le x\le E'$. Therefore, $\max(E+E'-m,0)\le  x\le \min(E,E')$.  Define the binomial distribution $B(k|n,p)=\binom{n}{k}p^k(1-p)^{n-k}$. Then
\bal{\label{pcond}
p(E'|E,d) &= \sum_{x=\max(E+E'-m,0)}^{\min(E,E')}B(x|E,p_1)B(E'-x|m-E,p_0)\,.
}
Finally, define
\bal{
\sum_{s,s'\ne s} \braket{\delta_{E_s,E}\delta_{E_{s'},E'}} &= \sum_{s,s'\ne s} p(E_s=E)p(E_{s'}=E'|E_s=E)\nonumber\\
&= Np(E)\sum_{d=1}^{n}\binom{n}{d} p(E'|E,d)\,.\nonumber
}
Substitution in Eq.~\cref{NENE} gives the desired result. 
\end{proof}

We compare the theoretical estimate of $\Sigma_{EE'}$ with a numerical estimate in Fig.~\ref{file15} for $n=15$, showing excellent agreement.
\begin{figure}[H]
\begin{center}
\includegraphics[width=0.6\textwidth]{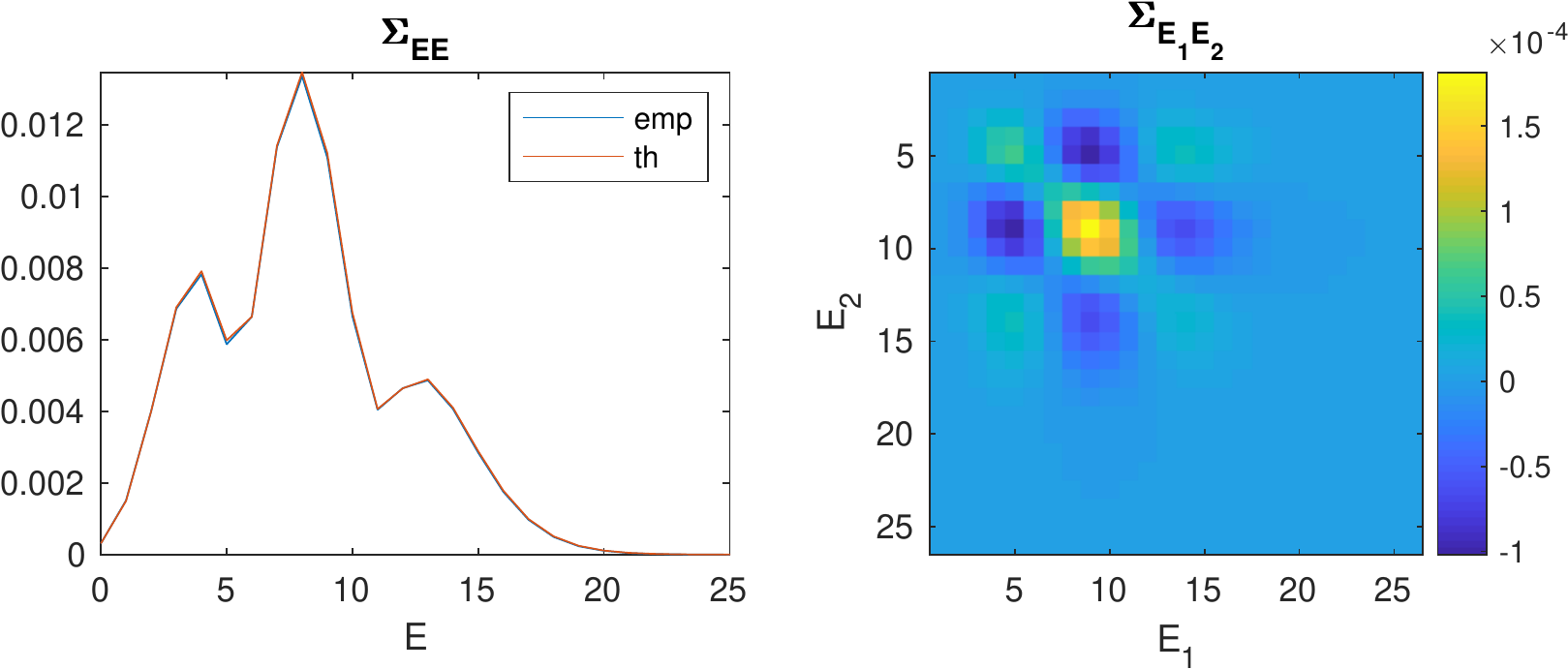}
\end{center}
\caption{Comparison of analytical estimate of 
$\Sigma_{EE'}^{(th)}= \braket{n_E n_{E'}} - p_E p_{E'}$ using Eqs.~\ref{binom} and~\cref{eq:nene} with numerical estimates $\Sigma_{EE'}^{(num)}$ using $10000$ random 3-SAT instances with $n=15, \alpha=4.2$.
{\bf Left}: Diagonal $\Sigma_{EE}$ versus $E$.
{\bf Right}: $\Sigma_{EE'}$ versus $E,E'$.
Maximal absolute error $\max_{E,E'}\left|\Sigma_{EE'}^{(th)}-\Sigma_{EE'}^{(num)}\right|=\num{2.9e-6}$.}
\label{file15}
\end{figure}
With Lemma~\ref{lemma2} we numerically compute the variance in $Z_1$ and $Z_2$
\baln{
\text{Var}( Z_1 ) &= \sum_{E,E'=1}^m\frac{\braket{n_En_{E'}}}{EE'}-\braket{Z_1}^2\,,
\qquad \text{Var}( Z_2 )= \sum_{E,E'=1}^m\frac{\braket{n_En_{E'}}}{E^2E'^2}-\braket{Z_2}^2\,.
}
We find numerically that $\sqrt{\text{Var}(Z_1)}$ scales as $\cO\left(n^{-5/2}\right)$ and $\sqrt{\text{Var}(Z_2)}$ scales as $\cO\left(n^{-7/2}\right)$ (see Fig.~\ref{file16b}).
\begin{figure}[H]
\begin{center}
\includegraphics[width=0.5\textwidth]{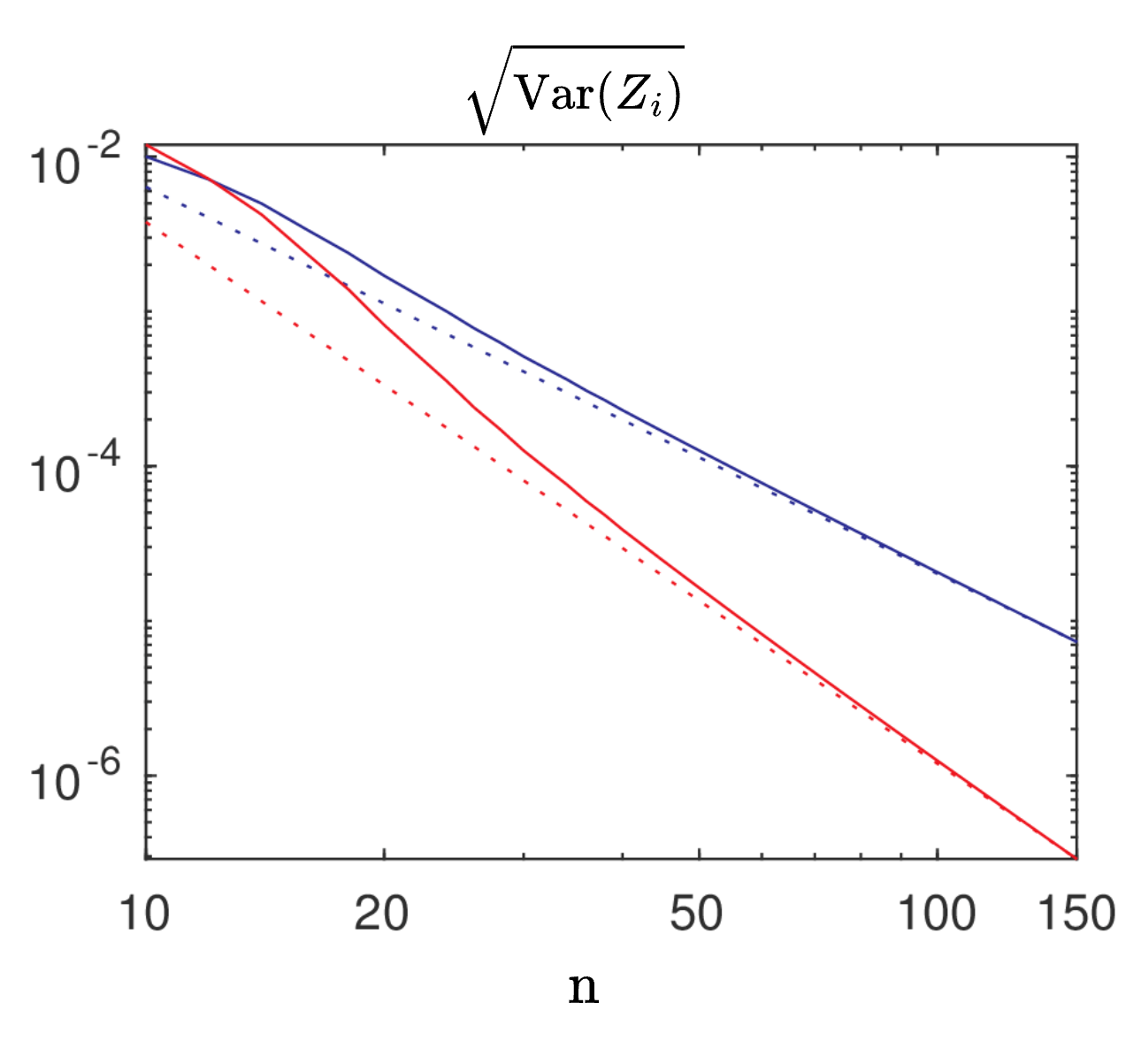}
\end{center}
\caption{Scaling of $\sqrt{\text{Var}(Z_1)}$ (blue) and and $\sqrt{\text{Var}(Z_2)}$ (red) versus $n$ for 3-SAT problem with  $\alpha=4.2$. Dotted lines are $n^{-5/2}$ (blue) and $n^{-7/2}$ (red).}
\label{file16b}
\end{figure}

\bibliography{references}

\end{document}